\documentclass{article}
\newif\ifarxiv\arxivtrue
\usepackage{booktabs} 
\usepackage{tikz} 
\usepackage[ruled]{algorithm2e} 
\renewcommand{\algorithmcfname}{ALGORITHM}
\SetAlFnt{\small}
\SetAlCapFnt{\small}
\SetAlCapNameFnt{\small}
\SetAlCapHSkip{0pt}
\IncMargin{-\parindent}
\usepackage{subcaption}
\usepackage{algorithmic}
\usepackage{enumitem}

\usepackage{hyperref}
\usepackage{url}
\usepackage{comment}
\usepackage{graphicx}
\usepackage{amsmath,amssymb}
\usepackage{amsthm}
\usepackage{fullpage}
\usepackage{bm}
\usepackage{xcolor}
\usepackage[square, comma, numbers,sort&compress]{natbib}
\usepackage{subcaption}
\usepackage{algorithmic}

\newtheorem{proposition}{Proposition}

\usepackage{sidecap}
\sidecaptionvpos{figure}{c}

\newcommand{\Zp}{\mathbb{Z}^+}
\newcommand{\lev}{\mathcal{L}}
\newcommand{\actvec}{\overline{\alpha}}
\newcommand{\Cinc}{\mathtt{Cost}^{\mathrm{inf}}}
\newcommand{\Cinchat}{\widehat{\mathtt{Cost}}^{\mathrm{inf}}}
\newcommand{\Cdec}{\mathtt{Cost}^{\mathrm{imp}}}
\newcommand{\Ii}{I^0}
\newcommand{\Infec}{\mathtt{Infect}}
\newcommand{\Poi}{\mathrm{\textsc{Poisson}}}
\newcommand{\scg}{\mathtt{SC}}
\newcommand{\I}{\mathtt{I}}
\newcommand{\Co}{\mathtt{Cost}}
\newcommand{\Cohat}{\widehat{\mathtt{Cost}}}
\newcommand{\Cimp}{\mathtt{Cost}^{\mathrm{imp}}}
\newcommand{\Cnc}{\mathtt{Cost}^{\mathrm{NC}}}
\renewcommand{\S}{\mathcal{S}}
\newcommand{\C}{\mathcal{C}}
\newcommand{\CCS}{\texttt{CCS}\xspace}
\newcommand{\CU}{\texttt{CU}\xspace}
\newcommand{\EQtwo}{\texttt{EQ2L}\xspace}
\newcommand{\EQthree}{\texttt{EQ3L}\xspace}
\newcommand{\II}{\texttt{Init}\xspace\texttt{Inf}\xspace}

\newtheorem{remark}{Remark}

\newcommand{\sj}[1]{\ifdraft{\color{blue}[Shahin: {#1}]}\fi}
\newcommand{\mc}[1]{\ifdraft{\color{magenta}[Mithun: {#1}]}\fi}
\newcommand{\fj}[1]{\ifdraft{\color{red}[Feiran: {#1}]}\fi}
\usepackage{fullpage}

\title{A Game-Theoretic Approach for Hierarchical Epidemic Control}
\author{
Feiran Jia, Aditya Mate, Zun Li, Shahin Jabbari, Mithun Chakraborty\\ Milind Tambe, Michael Wellman, Yevgeniy Vorobeychik
}

\begin{document}

\maketitle 
\thispagestyle{plain}
\pagestyle{plain}

\begin{abstract}
We design and analyze a multi-level game-theoretic model of hierarchical policy interventions for epidemic control, such as those in response to the COVID-19 pandemic. 
Our model captures the potentially mismatched priorities among a hierarchy of policy-makers (e.g., federal, state, and local governments) with respect to two cost components that have opposite dependence on the policy strength --- post-intervention infection rates and the socio-economic cost of policy implementation. 
Additionally, our model 
includes a crucial third factor in decisions: a cost of non-compliance with the policy-maker immediately above in the hierarchy, such as non-compliance of counties with state-level policies.
We propose two novel algorithms for approximating solutions to such games.
The first is based on best response dynamics (BRD), and exploits the tree structure of the game.
The second combines quadratic integer programming (QIP), which enables us to collapse the two lowest levels of the game, with best response dynamics.
Through extensive experiments, we show that our QIP-based approach significantly outperforms the BRD algorithm both in running time and the quality of equilibrium solutions.
Finally, we apply the QIP-based algorithm to experiments based on both synthetic and real-world data under various parameter configurations, and analyze the resulting (approximate) equilibria to gain insight into the impact of decentralization on overall welfare (measured as the negative sum of costs) as well as emergent properties like free-riding and fairness in cost distribution among policy-makers.
\end{abstract}

\section{Introduction}\label{sec:intro}
Democratic governments and institutions typically have a hierarchical structure. For example, policies in the U.S., Canada, and many European democracies emerge from complex interactions among the federal and state governments, as well as county boards, city councils and mayors. Such interactions are characterized by inherent asymmetries across different levels of the hierarchy. On the one hand, the specifics of policy formulation and enforcement (e.g., training and deployment of personnel and updating of infrastructure) are generally in the hands of administrative bodies at lower levels of the hierarchy --- often the \textit{lowest} level --- for practical reasons; actions these entities take are the ones that truly matter in the sense that they directly impact costs and benefits realized at all levels. On the other hand, entities at higher levels may have the power to impose constraints in some form or another on the policy-makers within their immediate jurisdiction (e.g., the U.S. federal government can constrain state policies);
violations of these constraints, in turn, entail a non-compliance cost to the violator, such as legal costs, penalties, or reputation loss.
Examples of such hierarchical policy structure arise in the spheres of education (e.g., topics to be included in primary education), healthcare (e.g., vaccination) and immigration. 
A preeminent recent example of such hierarchical policy-making is the response to the ongoing COVID-19 pandemic in countries with decentralized administration.
Policies concerning social distancing, masking and vaccination have involved recommendations at the federal level, guidelines and restrictions at the state/province/district level, and measures adopted by specific counties, cities or even individual businesses and schools.

In general, policies are contentious.
Agents at all levels of the policy-making hierarchy may disagree about the best policies, or more fundamentally, about the particular trade-offs made in devising policies e.g., lockdowns may have considerable economic and socio-psychological costs, but reduce the number of infections, and states and counties may disagree about how we should trade off these opposing considerations.

We propose a novel but stylized game-theoretic model that aims to capture the strategic tension in hierarchical policy-making for epidemic control. Our model is a game among policy-makers at 3 levels: Government at level 1, States at level 2 and Counties at level 3. In this game, policies at the higher level have an impact by imposing non-compliance costs on the level directly below, but ultimate implementation of policies happens by the Counties at the lowest level. Each agent trades off two types of costs: policy \emph{implementation cost} such as socio-psychological or economical impacts of lockdowns and \emph{infection cost} i.e., the number of new infections. Besides affecting the structure of utilities, the hierarchy also impacts the sequence of moves: agents at higher levels precede lower levels (e.g., by announcing guidelines), with the latter observing and reacting to the policy recommendations by the level directly above them. We stylize the actions to be single bounded scalars representing the degree of social constraints imposed in response to the epidemic (e.g., the strictness or extent of enforcement of mask or vaccination mandates).

Our first contribution is an analytic 
approximation of a recently proposed agent-based model (ABM) for COVID-19 pandemic spread estimation that accounts for social distancing~\citep{wilder2020modeling};
we show that our analytic model closely mirrors the short-term behavior of the ABM with a much smaller computational cost compared to the ABM.
Next, we propose two algorithmic approaches for approximating equilibrium solutions to our game.
Our first approach is a variation of best response dynamic that leverages the tree structure of the proposed game model; we term this approach \emph{BRD}. Best response dynamics is a common heuristic used for solving \emph{simultaneous-move games}~\citep{Fudenberg98}; our approach, in contrast, is designed to approximate subgame-perfect Nash equilibria of games with sequential structure.
Our second approach leverages our analytic estimation of the infection spread (and hence the utility functions) to collapse the lowest two layers into a quadratic integer program. It does this by combining a second-order Taylor approximation of utilities with KKT conditions for an equilibrium at the lowest level in the hierarchy.
We call the second approach Quadratic integer programming (\emph{QIP} for short). Through extensive experiments we show that \emph{QIP} significantly outperforms \emph{BRD} both in terms of running time and quality of equilibrium approximation.

Finally, we use our modeling  framework and proposed solution techniques to experimentally investigate possible phenomena arising from decentralized policy-making.  First, we investigate the impact of forms of (de)centralization that vary both qualitatively and quantitatively on the \textit{overall cost} of the impacted population, measured as a weighted sum of the infections and implementation costs. We show that decentralization is generally detrimental to social cost except for a small, non-obvious regime in which it can lower the social cost relative to a uniform centralized policy. Our next question of interest relates to \textit{policy free-riding}: Is it possible that (in equilibrium) a agent lower in the hierarchy adopts a weak policy with a low implementation cost while imposing a negative externality on another agent (perhaps on the same level) and hence also enjoying lower infection numbers owing to the latter agent's stronger policy? Can a higher-level policy-maker mitigate such free-riding via non-compliance penalties? We show that the answer depends in a complex manner on different parameters such as initial infection rates, degree of contact among different parts of the population, weights on different types of costs, and the non-compliance cost structure. Our final set of experiments measures the \textit{fairness} in the distribution of costs as a function of model parameters as well as degrees of centralization. Our experimental results are reported on synthetic data as well as  real-world data based on recent transportation and census data for two U.S. States, New York and New Jersey.

\subsection{Related work}\label{sec:relwork}
Our work is related to the line of research applying the social and behavioral sciences to the cost-benefit analysis of both centralized and decentralized decision-making under pandemic/epidemic conditions \citep{van2020using}. Some papers approach these trade-offs from an \textit{optimal control} perspective \citep{sethi1978optimal,rowthorn2012optimal,fenichel2013economic,rowthorn2020cost}. Others study the equilibria of various game-theoretic models of \textit{individuals} deciding whether to follow guidelines for preventive measures (distancing, vaccination, etc.) and treatment, possibly against the (perceived) aggregate behavior of the population, under various models of disease propagation; e.g. the differential game model \citep{reluga2010game}, the ``wait and see" model of vaccinating behavior \citep{BHATTACHARYYA20115519}, evolutionary game-theoretic models \citep{kabir2020evolutionary,brune2020evolutionary}, and various others \citep{chen2011public,chen2012mathematical,gersovitz2010disinhibition,toxvaerd2019rational} (see, e.g. \citep{toxvaerd2020equilibrium} for a summary). Unlike these works, we model the strategic interactions among \textit{ideologically diverse, hierarchical policymakers} with explicit \textit{non-compliance penalties}, and experimentally assess the impact of such interactions upon the actually implemented policies under various parameter settings. 

Also related is the literature on ABM for pandemic spread and response policies that account for preferences/incentives of individuals \citep{fenichel2011adaptive,wilder2020modeling,hoertel2020stochastic}. Wilder et al.~\citep{wilder2020modeling} is of particular importance since our policy impact cost is computed by a closed-form approximation to their model. Other recent work includes the assessment of the impact of prevention and containment policies on the spread COVID-19 via causal analysis \citep{mastakouri2020causal}, Gaussian processes \citep{qian2020and}, and state-of-the-art data-driven non-pharmaceutical intervention models \citep{sharma2020robust}.

Instead of an analytic treatment, we empirically compute the (approximate) equilibrium of our complex, multi-level, continuous-action game using algorithmic approaches that exploit the structure of the problem. Thus, our methods belong to the category of \textit{empirical game-theoretic analysis} (see, e.g. \citep{wellman2006methods,vorobeychik2007constrained,vorobeychik2008stochastic,gatti2011equilibrium}).
Finally, Li et al.~\citep{LiJMJCTV21} have recently introduced a general model for hierarchical games. However, their utility structure is different and our goal is to compute global equilibria as opposed to local ones. 

\section{Hierarchical Epidemic Control}\label{sec:model}
\subsection{The Game Model}\label{sec:game_model}
\begin{figure}[ht!]
	\centering
	\includegraphics[scale=0.3]{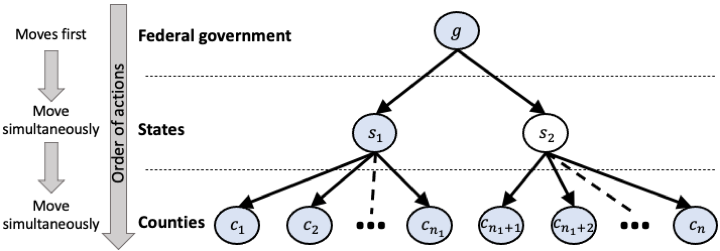}
	\caption{The hierarchy in an HECG. Arrow determines the unique parent of each node. 
	The cost of each node depends on a subset of nodes e.g. the cost of  State $s_1$ can have functional dependency on the shaded nodes only. }
	\label{fig:model}
\end{figure}
We represent the \textit{players} in the \emph{hierarchical epidemic control game (HECG)} by \textit{nodes} in a directed rooted tree, as in Figure~\ref{fig:model}. 
An \emph{HECG} has a $3$-level structure:  Government $g$ in the first level, States $\S = \{s_1, \ldots, s_m
\}$ in the second, and Counties $\C = \{c_1, \ldots, c_n\}$ in the third.  The Counties are partitioned among States and, hence, each County has a unique State as its parent.
Each player $a \in g \cup \S \cup \C$ can take a scalar \textit{action} $\alpha_a \in [0,1]$ which abstracts the policy adopted by $a$, capturing the \textit{the extent of overall activity} implemented/recommended by $a$. Conversely, $1-\alpha_a$ represents the extent of social distancing. Thus, small $\alpha_a$ corresponds to the greatest reduction in infection spread (due to stricter social distancing). On the other hand, large $\alpha_a$ entails a higher policy implementation cost, such as socio-economic and psychological costs of social distancing. At the extremes, $\alpha_a=1$ signifies no intervention, while $\alpha_a=0$ corresponds to a complete lockdown. 

\emph{HECG} is a sequential game in which players make strategic decisions following the sequence of layers.
Specifically, the Government in level 1 moves (i.e., chooses a strategy) first, followed by all the States in level 2, who first observe the action of the Government and simultaneously choose a joint strategy profile in response. This is then followed by all the Counties in level 3 who observe the strategy of the States and simultaneously respond with joint strategy. See Figure~\ref{fig:model}.

All the utilities for COVID-19 social distancing policies are negative (i.e., costs), which we define next. Let $\actvec$ denote the profile of actions of all players. The cost function of each player $a$ has three components. \emph{(i) Infection cost}, $\Cinc_a(\actvec)$, which is a measure of infection spread (number of people infected in the player's geographic area, say). \emph{(ii) Implementation cost}, $\Cdec_a(\actvec)$, that models the psychological and economic costs of a lockdown. \emph{(iii)} For each player except the Government, \emph{non-compliance cost}, $\Cnc_a(\alpha_a,\alpha_{\pi(a)})$ where $\pi(a)$ denotes the parent of player $a$, which is a penalty imposed by a policy-maker upon an agent within its jurisdiction for deviating from its recommendation (e.g., a fine, litigation costs, or reputational harm). An important piece of structural feature to the policy implementation and impact costs is that they directly depend only on the actions of the Counties since the actions of the Government and the States are merely recommendations whereas the actual social distancing policies would be implemented based on the action of the Counties. 

To formalize, we introduce for each player $a$ the notion of its \textit{share} $\mu_a \in [0,1]$ which can be interpreted as the proportion of the total population of the country that is under the jurisdiction of the corresponding player (e.g., the share of a State is the proportion of the total population of the Counties in that state). Thus, $\mu_g = 1$. For a State $a \in \S$, we have $\mu_a=\Sigma_{a' \in \chi(a)} \mu_{a'}$ where $\chi(a)$ denote the set of children of $a$ i.e., the Counties under the jurisdiction of $a$.
The shares of the nodes in the Counties are arbitrary, except for the constraint $\Sigma_{a \in \C}\mu_a=1$.
We now use the notion of shares to formally define the infection and implementation costs of policies.
\begin{itemize}
	\item For the Counties $a \in \C$, the infection cost $\Cinc_a(\actvec)$ depends only on $\actvec_\C$, the joint action of all Counties. Furthermore, this cost lies in $[0,1]$, and is non-decreasing in each $\alpha_a \in \actvec_\C$; we provide further specifics of this function in Section~\ref{sec:infec}. For a higher-level player such as the Government or a State, this cost is the share-weighted aggregate of those of its child-nodes: $$\Cinc_a(\actvec) = \frac{1}{\mu_a}\sum_{a' \in \chi(a)} \mu_{a'}\Cinc_{a'}(\actvec)$$ for $a\in g\cup \S$.
	\item For the counties $a \in \C$, the implementation cost is simply $\Cdec_a(\actvec)=1-\alpha_a$; so the cost decreases as $\alpha$ increases and is bounded in $[0,1]$. For a higher-level player such as the Government or a State, this cost is the share-weighted aggregate of those of its child-nodes:
	$$\Cdec_a(\actvec) = \frac{1}{\mu_a}\sum_{a' \in \chi(a)} \mu_{a'}\Cdec_{a'}(\actvec).$$
	\item For the non-compliance cost, any deviation from the recommended policy of a player's parent is penalized~\citep{nytimes}, with the discrepancy being measured by the Euclidean distance i.e., for all players $a\in \S\cup \C$, $$\Cnc_a(\alpha_a,\alpha_{\pi(a)})= \left(\alpha-\alpha_{\pi(a)}\right)^2.$$ The Government incurs no non-compliance cost.
\end{itemize}

Finally, each player $a \in \S \cup \C$ has an idiosyncratic set of  \textit{weights} $\kappa_a \geq 0$ and $\eta_a\geq 0$ that trade its three cost components against each other via a convex combination and account for differences in ideology. 
Thus, the overall cost of such a player $a$ is given by
\begin{align*}
	\Co_a(\actvec) := \kappa_a \Cinc_a(\actvec) + \eta_a \Cdec_a(\actvec) + \gamma_a \Cnc_a(\alpha_a,\alpha_{\pi(a)}),
\end{align*}
where $\gamma_a = 1 -\kappa_a-\eta_a$. The Government has only one weight $\kappa_{g} > 0$, its overall cost being
\begin{align*}
\Co_{g}(\actvec) := \kappa_{g} \Cinc_{g}(\actvec) + (1-\kappa_{g})\Cdec_{g}(\actvec).
\end{align*}

\subsection{Solution Concept}\label{sec:eqbm_select}
The solution concept we are primarily interested in is a \textit{pure-strategy subgame perfect Nash equilibrium} (PSPNE) \citep{shoham2008multiagent} of our continuous-action game which is sequential-move between levels and simultaneous-move within a level. In general, a PSPNE may not exist.
Consequently, we will seek to compute an $\epsilon$-PSPNE, where $\epsilon$ is the highest benefit from deviation by any player $a$.
In Section~\ref{sec:methods}, we present a general approach for finding such approximate equilibria in our setting.

\begin{remark}\label{compli_setting}
\normalfont
As a special case of our model, we consider a setting where all Counties comply with the policy recommended by their respective States i.e., $\lambda_{c}=1$ for all $c\in\C$. We call this the \emph{compliant} setting. In this setting, each County has a dominant strategy and this reduces the game to a two-level Stackelberg game with one leader (Government) and multiple followers (States). A PSPNE in this case corresponds to a Stackelberg equilibrium.
\end{remark}

\subsection{Infection Dynamics and Cost}\label{sec:infec}
We now come to the computation of $\Cinc_a(\cdot)$ for each County. Recently, Wilder et al.~\citep{wilder2020modeling} developed and analyzed an agent-based model (ABM) for COVID-19 spread 
that accounts for the \textit{degree of contact} (both within and between households) among individuals from different parts of a population.\footnote{The model in~Wilder et al.~\citep{wilder2020modeling} is an individual-level variant of the well-known susceptible-exposed-infectious-recovered or SEIR model but we assume that every exposed person eventually becomes infected after an incubation period.} 
However, this ABM is computationally expensive, making its use for equilibrium computation impractical at scale. In this section, we will derive a closed-form model of infection spread that (as we show below) relatively closely mirrors the expected number of infections of the ABM over a short horizon. 

Let $N_a$ and $\Ii_a$ denote the fixed population of County $a$ and the number of infections in $a$ before policy intervention respectively. An individual who is not currently infected but can develop an infection on contact with someone infected is \textit{susceptible}. We call an individual from County $a'$ \textit{active} in County $a$ if that individual is capable of making contact (through travel etc.) with a susceptible individual in County $a$; if $a'=a$, we say that the individual is active within County $a$. A major parameter of the ABM is the \textit{transport matrix} $R = \{r_{aa'}\}_{a,a' \in \C}$, where $r_{aa'} \ge 0$ is the proportion of the population of County $a'$ that is active in County $a$ in the absence of an intervention. Thus, in the absence of policy intervention, the total number of individuals from County $a'$ active in County $a$ is $N_{a'}r_{aa'}$ and the total number of infected individuals from County $a'$ active in County $a$ is $\Ii_{a'}r_{aa'}$.

The policy $\alpha_{a}$ affects the population in two ways: it scales down both the susceptible and active sub-populations. In other words, under policy intervention, County $a$ has $(N_a - \Ii_a)\alpha_a$ susceptible individuals, and there are $N_{a'}\alpha_{a'}r_{aa'}$ active individuals in County $a$ from County $a'$, out of whom $\Ii_{a'}\alpha_{a'}r_{aa'}$ are (initially) infected. Hence, the proportion of infected active individuals in County $a$ is
$$\rho_a(\actvec_\C) := \frac{\sum_{a' \in \C}\Ii_{a'} \alpha_{a'} r_{aa'}}{\sum_{a' \in \C}N_{a'} \alpha_{a'} r_{aa'}}.$$

We will now focus on an arbitrary susceptible individual in County $a$ and lay down our assumptions on the process why which she may contract an infection: This individual makes actual contact with a random sample of $X$ active individuals drawn from a Poisson distribution with mean $C$, which is a parameter in our model \citep{wilder2020modeling}; this distribution is fixed across all individuals in all Counties, and all these contacts are mutually independent. The next assumption is that, in this sample of $X$ contacts for a susceptible individual in County $a$, the proportion of infected individuals is $\rho_a(\actvec_{\C})$.

Let $p \in (0,1)$ denote the probability that a susceptible individual becomes infected upon contact with an infected individual, i.e. the probability that contact with an infected individual does not infect a susceptible individual is $(1-p)$. Since all $X\rho_a(\actvec_\C)$ infected contacts of an arbitrary susceptible individual are mutually independent, the probability that the susceptible individual develops an infection is $1-(1-p)^{X\rho_a(\actvec_\C)}$. We also interpret this as the proportion of the $(N_a - \Ii_a)\alpha_a$ susceptible individuals in County $a$ who end up getting infected. Let $\Infec_a(\actvec_\C)$ denote the \textit{expected} number of 
\textit{additional, post-intervention infections} in County $a$. Thus,
\begin{align*}
\Infec_a(\actvec_\C) &= \mathbb{E}_X [(N_a - \Ii_a)\alpha_a(1-(1-p)^{X\rho_a(\actvec_\C)})]\\&= (N_a - \Ii_a)\alpha_a(1-\mathbb{E}_X [((1-p)^{\rho_a(\actvec_\C)})^X]).
\end{align*}
Define $y_a(\actvec_\C) := (1-p)^{\rho_a(\actvec_\C)}$. Since $X \sim \Poi(C)$, Proposition~\ref{prop:exppoisson} in Appendix~\ref{app:closedform} tells us that 
\begin{align}
&\Infec_a(\actvec_C) = (N_a - \Ii_a)\alpha_{a}(1-e^{-C(1-y_a(\actvec_\C))}). \label{closedform}
\end{align}
Finally, the infection cost is $\Cinc_a(\actvec):=\Infec_a(\actvec_C)/N_a$.

\begin{figure}[ht!]
\centering
	\begin{subfigure}[t]{0.49\columnwidth}
	\centering
		\includegraphics[width=.99\columnwidth]{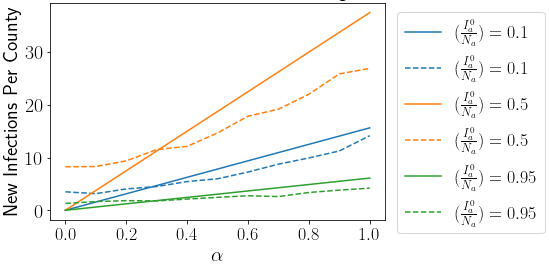}
		\caption{New infections vs. policy (shared by Counties).}
		\label{fig:closedalpha}	
	\end{subfigure}~
	\begin{subfigure}[t]{0.49\columnwidth}
	\centering
		\includegraphics[width=.99\columnwidth]{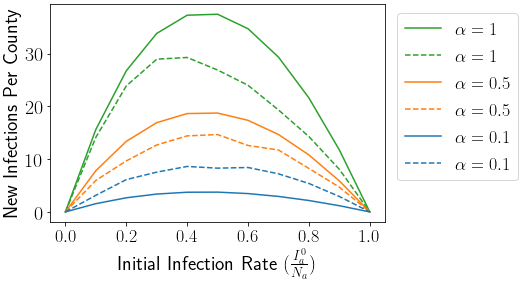}
		\caption{New infections vs. initial infection (shared by Counties).}
		\label{fig:closedinitFrac}
	\end{subfigure}
	\caption{Comparison of ABM output  (dashed lines) with closed-form approximation (solid lines). \label{fig:ABM_vs_closedform}}
\end{figure}
We ran some preliminary experiments comparing Equation~\eqref{closedform} with the actual output of the ABM~\citep{wilder2020modeling}; partial results are shown in Figure~\ref{fig:ABM_vs_closedform}. 
Note that Equation~\eqref{closedform} is a one-shot formula, whereas the ABM computes contacts and infections recursively over several time-periods with an initial \textit{incubation period} so that the effect of the first-period contacts are manifested only after a delay. Hence, we contrast the ABM output after $8$ periods (to account for the average incubation period of 7 days~\citep{lauer2020incubation}) with the above closed-form estimation. 
In the experiments we report, we have $2$ States under the Government, each State having $2$ Counties; each County $a$ has a population of $N_a=250$; the transport matrix is symmetric, given by $r_{aa'} = 0.25$ for every pair of Counties $a,a'$. 
We set $p = 0.047$ \citep{wilder2020modeling} and $C=15$ (calculated based on Prem et al.~\citep{prem2017projecting}).\footnote{In this symmetric setting, the number of new infections grows linearly with $\alpha$.}
For each set of experiments (represented by a separate color in Figure~\ref{fig:ABM_vs_closedform}), each County has the same \textit{initial infection rate} $\Ii_a/N_a$ and applies the same policy $\alpha$.
In Figure~\ref{fig:closedalpha}, we vary $\alpha$ on the x-axis, for different (fixed) values of $\Ii_a/N_a$; similarly, in Figure~\ref{fig:closedinitFrac}, for different policies, we vary the initial infection rates. The plots indicate qualitative similarity between the ABM and our approximation; a salient point of similarity is that the additional number of infections decreases as the initial infection rate gets higher or lower than a middling point, everything else remaining the same. This is because a higher infection rate implies less ``room for growth'' due to a fixed population, whereas a lower value of the same rate causes fewer further infections over the same horizon.

\section{Solution Approach}\label{sec:methods}
We now propose two methods to compute approximate PSNE for HECG. Our first approach, in Section~\ref{sec:brd}, is based on backward induction and best response dynamics which are common approaches for sequential games. Our second approach, in Section~\ref{sec:qip-methods}, speeds up the equilibrium computation compared to our first approach with only minor degradation on solution quality.
\subsection{Best Response Dynamics (BRD)}
\label{sec:brd}
HECG is a sequential game model with one-dimensional action space for each agent resulting in a non-convex strategic landscape.
To seek for PSPNE, we propose a backward induction algorithm incorporating a payoff point query interface and a best response computation component solving for a  joint profile in equilibrium.
The algorithm exploits the hierarchical structure by propagating strategic information between consecutive levels:
given the action of the Government, the States constitute a simultaneous-move game whose payoffs emerge from the strategic interactions of the Counties.
To obtain payoffs for the Government and States, we recursively call to the next level till we reach the County players.
Then we use these payoffs to solve for an approximate PSNE.
Since every such simultaneous-move game lacks the tractable analytic payoff structure for gradient-based optimization, in our current implementation we discretize the strategy space and adopt \emph{best response dynamics} (BRD) for equilibrium computation.

\ifarxiv
\begin{algorithm}[ht!]
\caption{BRD at level $l \in \{1, 2, 3\}
$}
\label{alg:ne}
\textbf{Input}: $\actvec_1, \ldots, \actvec_{l-1}$ (actions of players in levels $1$ to $l-1$), 
$T=\{T_1, T_2, T_3\}$ (rounds of best response before termination for each level), $e=\{e_1, e_2, e_3\}$ (approximation error tolerance in each level).\\
\begin{algorithmic}[1] 
\STATE Initialize $t\leftarrow 1$, $\epsilon_l \leftarrow \infty$ and $\actvec_l$ randomly.
\WHILE{$t \leq T_l$ or $\epsilon_l \leq e_l$}
\FOR{Any player $a$ in level $l$}
\IF{$l = 3 $ i.e., $l$ is the lowest level}

\STATE $\alpha_{a} \leftarrow \arg \min_{\alpha'_{a}}\Co_{a}(\actvec)$.
\ELSE
\STATE $\alpha_{a} \leftarrow \arg \min_{\alpha'_{a}}\Co_{a}( \text{BRD}(\actvec_1, \ldots, \actvec_{l}, T, e)$.
\ENDIF
\ENDFOR
\STATE Calculate the $\epsilon_l$ for profile $\actvec_l$. 
\STATE Update $\epsilon_l$ if lower than the current value. 
\STATE $t \leftarrow t+1$.
\ENDWHILE
\STATE \textbf{return} $\actvec^*$ where $\actvec^*_l$ has the lowest $\epsilon_l$.
\end{algorithmic}
\end{algorithm}
\else
\begin{algorithm}[ht!]
\caption{BRD at level $l \in \{1, 2, 3\}
$}
\label{alg:ne}
\textbf{Input}: $\actvec_1, \ldots, \actvec_{l-1}$ (actions of players in levels $1$ to $l-1$), 
$T=\{T_1, T_2, T_3\}$ (rounds of best response before termination for each level), $e=\{e_1, e_2, e_3\}$ (approximation error tolerance in each level).\\
\begin{algorithmic}[1] 
\State Initialize $t\leftarrow 1$, $\epsilon_l \leftarrow \infty$ and $\actvec_l$ randomly.
\While{$t \leq T_l$ or $\epsilon_l \leq e_l$}
\For{Any player $a$ in level $l$}
\If{$l = 3 $ i.e., $l$ is the lowest level}

\State $\alpha_{a} \leftarrow \arg \min_{\alpha'_{a}}\Co_{a}(\actvec)$.
\ELSE
\State $\alpha_{a} \leftarrow \arg \min_{\alpha'_{a}}\Co_{a}( \text{BRD}(\actvec_1, \ldots, \actvec_{l}, T, e)$.
\EndIf
\EndFor
\State Calculate the $\epsilon_l$ for profile $\actvec_l$. 
\State Update $\epsilon_l$ if lower than the current value. 
\State $t \leftarrow t+1$.
\EndWhile
\State \textbf{return} $\actvec^*$ where $\actvec^*_l$ has the lowest $\epsilon_l$.
\end{algorithmic}
\end{algorithm}
\fi

Algorithm~\ref{alg:ne} computes the $\epsilon$-equilibrium among players at a single level given tunable parameters. For $l \in \{1, 2, 3\}$, let $T_l$ denote the pre-determined cap on the number of steps of BRD at each level $l$; define $\epsilon_l$ as the largest improvement in payoff computed for any player at level $l$ due to a deviation from its current strategy to its best response to the current profile (see the details of best response computation in the next paragraph), and let $e_l$ the fixed upper limit on the $\epsilon_l$ for all rounds. 
To prevent BRD from getting trapped in a cycle, we maintain a history of the strategy profiles generated by BRD and check whether a new profile already exists in the memory; if yes (i.e. a cycle is detected), we resume by generating a random profile.
Finally, the algorithm returns the profile with the lowest $\epsilon_l$ upon termination. 

\begin{figure}[ht!]
  \centering
  \includegraphics[scale=0.5]{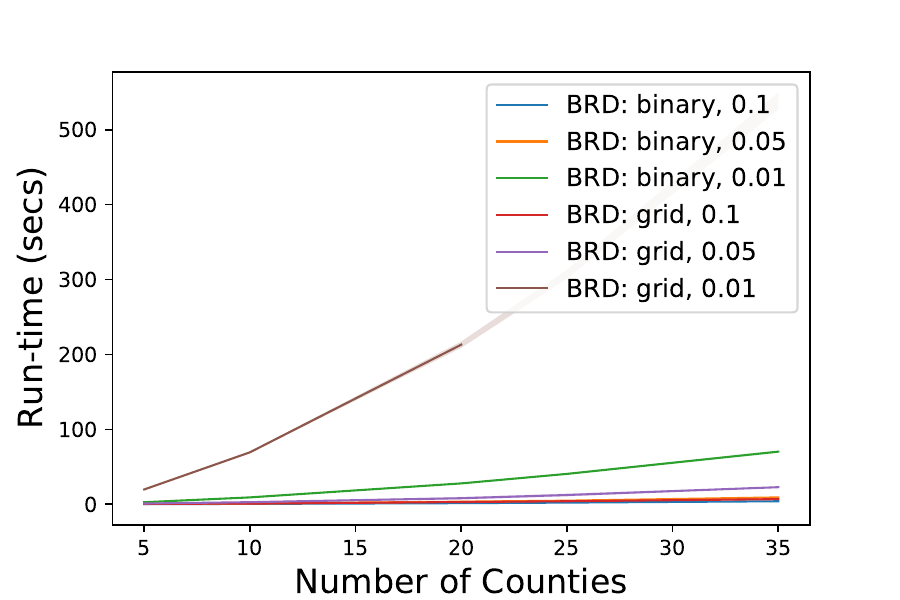}
  \caption{Run-time (y-axis) of binary grid search versus the number of States (x-axis). The number ($0.1,0.05,0.01$) next to the name of the method (binary/grid) indicates the grid size.}
  \label{fig:run_time_grid_binary}
\end{figure}

To search for the best-response strategy (an optimization problem), we discretized the continuous strategy space of each agent, i.e., the interval $[0,1]$, and used grid search with tie-breaking (smaller policy impact) to recover the optimum value. 
We also investigated empirically whether we could use \textit{binary search} (i.e. the bisection method) for best response computation at the lowest level of the game to speed up BRD with low loss in accuracy, although there is no guarantee since the overall cost of a lowest-level player is not convex (but could still be ``almost'' convex).
Figure~\ref{fig:run_time_grid_binary} shows run-times of Algorithm~\ref{alg:ne} in the compliant, symmetric setting (i.e., a two-level game described in Section~\ref{sec:eqbm_select}, with equal populations and a symmetric transport matrix as in the experiments of Section~\ref{sec:infec}) as the function of the number of States ranging from $10$ to $100$ when we replace the grid search in the second level with binary search. In all these comparison experiments, grid search finds PSPNE with $\epsilon_2 = 0$; binary search yields the same results (i.e., finds an exact equilibrium), but is more efficient.

\subsection{Quadratic Integer Programming (QIP)}\label{sec:qip-methods}
The non-trivial dependence of the overall cost of an agent on the actions of all Counties, due to the infection cost in Equation~\eqref{closedform}, prevents us from using standard optimization techniques for best response computation. A natural question is whether we can use a high-quality approximation to the cost function which will enable to us to invoke analytical methods to speed up the search for an (approximate) equilibrium even further. 

With this in mind, we derive a second-order Taylor approximation of Equation~\eqref{closedform} for each State --- this makes the overall cost quadratic since the implementation and non-compliance costs are linear and quadratic, respectively.
Given this observation, we develop an approach for computing the best response of each State (fixing the action of the Government and other States) under this (approximate) quadratic cost function using a quadratic integer program (QIP-Best-Response) by analyzing the KKT conditions and applying linearization at the Counties, as follows.

Recall that $\actvec_\C=(\alpha_{c_1}, \ldots, \alpha_{c_n})$ denotes the action profile of Counties. In this section, we will slightly abuse notation and treat $\actvec_\C$ as a column vector. For any player $a\in \C$, the infection cost $\Cinc_a(\actvec_\C)$ can be written as stated in Equation~\eqref{closedform}; for any County $a$, $\Cinc_a(\actvec_\C)$ has a non-trivial non-convex dependence on $\alpha_a$.

Let $\Cinchat_a(\actvec_\C, \actvec_{{\C}_0})$ denote the Taylor expansion of $\Cinc_a(\actvec_\C)$ around point $\actvec_{{\C}_0}$ up to the second-order term, i.e.,
\begin{align*}
    \Cinchat_a(\actvec_\C, \actvec_{{\C}_0})=\Cinc_a(\actvec_{{\C}_0})
    &+ \nabla \Cinc_a(\actvec_{{\C}_0})^\top (\actvec_{\C} - \actvec_{{\C}_0})\\
    &+ \frac{1}{2} (\actvec-\actvec_0)^\top
    \nabla^2 \Cinc_a(\actvec_{{\C}_0})(\actvec-\actvec_{\C_0}).
\end{align*}
Then, the total cost of County $a$ can be approximated by replacing $\Cinc_a(\actvec_\C)$ in $\Co_a(\actvec)$ with $\Cinchat_a(\actvec_\C, \actvec_{{\C}_0})$. Let us call this quantity $\Cohat_a(\actvec)$. The best response of County $a$ can be written as 
$$\alpha^*_a = \text{argmin}_{\alpha_a}\Cohat_a(\actvec), \quad \text{s.t. }\alpha_a\in [0,1].$$
We denote $G_a = \nabla \Cinc_a(\actvec_{{\C}_0})$, and $H_a = \nabla^2 \Cinc_a(\actvec_{{\C}_0})$. Here, for each agent $a$, $G_a$ (resp. $H_a$) is a vector (resp. matrix) indexed by all agents (resp. all pairs of agents). Let $G_a[a']$ (resp. $H_a[a',a'']$) denote the element corresponding to the agent $a'$ (resp. the ordered agent-pair $(a',a'')$), and $H_a[a',\cdot]$ the row corresponding to agent $a'$. 

Since the implementation and non-compliance cost components are unchanged in this approximation (linear and quadratic respectively), their first two derivatives can be easily computed and added to the above to obtain the first two derivatives of the overall function $\Cohat_a(\actvec)$.
We start by writing the Lagrange dual of the quadratic program:
$$L_a(\alpha_a, \beta_a, \rho_a)=\Cohat_a(\actvec)-\beta_a(\alpha_a)-\rho_a(1-\alpha_a).$$

By KKT conditions we have:
\begin{align*}
\beta_a \alpha_a & = 0, 
\rho_a (1-\alpha_a)  = 0, \\
\alpha_a &\leq 1, 
-\alpha_a \leq 0, \\
\beta_a &\geq 0, 
\rho_a  \geq 0.
\end{align*}

and the stationary condition:
\begin{align} \label{cos1} 
    \partial L_a/\partial \alpha_a
    &= \partial \Cinchat_a(\actvec_\C, \actvec_{{\C}_0})/\partial \alpha_a
    - \beta_a + \rho_a = 0,\\
    \text{where } \partial \Cinchat_a(\actvec_\C, \actvec_{{\C}_0})/\partial \alpha_a &= \kappa_{a} \big[ G_{a}[{a}] + H_{a}[{a}, \cdot](\actvec_{\C}-\actvec_{{\C}_0}) \big] \notag\\
    &\qquad - \eta_{a} + 2(1-\kappa_{a}-\eta_{a})(\alpha_{a}-\alpha_s). \notag
\end{align}

To linearize the complementary slackness conditions, we introduce two binary variables $\rho_{b,a}, \beta_{b,a}$, and a large number $M$ for each county. Then the conditions can be linearized as the following constraints 

\begin{align}
    1-\alpha_a &\leq \rho_{b,a}, \\ 
    \rho_a &\leq M(1-\rho_{b,a}),\\
    \alpha_a &\leq 1-\beta_{b,a}, \\
    \beta_a &\leq M\cdot \beta_{b,a},\\
    \quad &\rho_{b,a}, \beta_{b,a} \in \{1, 0\}. \label{cosN}
\end{align}

\ifarxiv
\begin{algorithm}[ht!]
\caption{Taylor-Iter-BR}
\label{alg:taylor_iter}
\textbf{Input}: 
$T$ (rounds of repeating the approximation), $\actvec_{\C_0}$ (County-level action profile around which infection cost is approximated). $a\in \S$ (State whose best response is computed), $\actvec_\S$ (action profile of all States).\\
\begin{algorithmic}[1] 
\STATE Initialize $t \leftarrow 0$.
\WHILE{$t \leq T$}
\STATE $\alpha_a,  \actvec_{\C} \leftarrow \text{QIP-Best-Response}(\actvec_\S, \actvec_{\C_0}, a)$. \label{qip_step}
\STATE $\actvec_{\C_0} \leftarrow  \actvec_{\C}$. 
\STATE $t\leftarrow t+1$.
\ENDWHILE
\STATE return $\alpha_a$, $\actvec_{\C_0}$.
\end{algorithmic}
\end{algorithm}
\else
\begin{algorithm}[ht!]
\caption{Taylor-Iter-BR}
\label{alg:taylor_iter}
\textbf{Input}: 
$T$ (rounds of repeating the approximation), $\actvec_{\C_0}$ (County-level action profile around which infection cost is approximated). $a\in \S$ (State whose best response is computed), $\actvec_\S$ (action profile of all States).\\
\begin{algorithmic}[1] 
\State Initialize $t \leftarrow 0$.
\While{$t \leq T$}
\State $\alpha_a,  \actvec_{\C} \leftarrow \text{QIP-Best-Response}(\actvec_\S, \actvec_{\C_0}, a)$. \label{qip_step}
\State $\actvec_{\C_0} \leftarrow  \actvec_{\C}$. 
\State $t\leftarrow t+1$.
\EndWhile
\State return $\alpha_a$, $\actvec_{\C_0}$.
\end{algorithmic}
\end{algorithm}
\fi

Thus, we can solve for the equilibrium in the County-level by solving the system of above constraints for each county $a$ when the State's strategies are given.
Then, the best-response of a State $s\in\S$ can be formulated as a bi-level optimization, which we call QIP-Best-Response. Let $\Cohat_s(\actvec, \actvec'_\C)$ denote the total cost of State $s$ when replacing each infection term for the Counties $a\in\C$ by $\Cinchat_a(\actvec_\C, \alpha'_a)$.
\begin{align*}
&\alpha^*_s = \text{argmin}_{\alpha_s, \{\alpha_a, \beta_a, \rho_a\}}\Cohat_s(\actvec, \actvec'_\C), \quad\\ &\text{s.t. }\alpha_s, \alpha_a \in [0,1]; \text{constraints } (\ref{cos1})-(\ref{cosN}), \forall a.  
\end{align*}

\ifarxiv
\begin{algorithm}[ht!]
\caption{QIP}
\label{alg:qip-brd}
\textbf{Input}: 
$\Delta_g$ (discretization factor for Government's action), $T_2$ (rounds of best-response before termination for the States), $k_2$ (number of States best responding), $e_2$ (approximation error tolerance for States), $\actvec_{\C_0}$ (County-level action profile around which infection cost is approximated).\\
\begin{algorithmic}[1] 
\STATE Initialize $\alpha'_g\leftarrow 0$, $\Co_g\leftarrow \infty$ and $\actvec \leftarrow \vec{0}.$
\WHILE{$\alpha'_g \leq 1$}
\STATE Initialize $\actvec_\S$ randomly, $t\leftarrow 0$ and
$\epsilon_2\leftarrow \infty$.
\WHILE{$t \leq T_2$ or $\epsilon_2 \leq e_2$}
\STATE Select $k_2$ players from $\S$.
\FOR{Each selected player $a$}
\STATE $\alpha'_a ,  \actvec_{\C}\leftarrow \text{Taylor-Iter-BR}(\actvec_\S, \actvec_{\C_0}, a)$. 
\ENDFOR
\FOR{Each selected player $a$}
\STATE Update $\actvec_\S$ by setting $\alpha'_a$ instead of $\alpha_a$.
\ENDFOR
\STATE $t\gets t+1$.
\STATE Calculate $\epsilon_2$ for $\alpha_\S$.
\STATE Update $\epsilon_2$ if lower than the current value.
\ENDWHILE
\STATE Compute $\Co_g(\actvec)$.
\STATE If lower than the current value update $\Co_g$ and set $\actvec = <\alpha'_g, \actvec_\S, \actvec_\C>$. 
\STATE $\alpha'_g\gets \alpha'_g+\Delta_g$.
\ENDWHILE
\STATE return $\actvec$.
\end{algorithmic}
\end{algorithm}
\else
\begin{algorithm}[ht!]
\caption{QIP}
\label{alg:qip-brd}
\textbf{Input}: 
$\Delta_g$ (discretization factor for Government's action), $T_2$ (rounds of best-response before termination for the States), $k_2$ (number of States best responding), $e_2$ (approximation error tolerance for States), $\actvec_{\C_0}$ (County-level action profile around which infection cost is approximated).\\
\begin{algorithmic}[1] 
\State Initialize $\alpha'_g\leftarrow 0$, $\Co_g\leftarrow \infty$ and $\actvec \leftarrow \vec{0}.$
\While{$\alpha'_g \leq 1$}
\State Initialize $\actvec_\S$ randomly, $t\leftarrow 0$ and
$\epsilon_2\leftarrow \infty$.
\While{$t \leq T_2$ or $\epsilon_2 \leq e_2$}
\State Select $k_2$ players from $\S$.
\For{Each selected player $a$}
\State $\alpha'_a ,  \actvec_{\C}\leftarrow \text{Taylor-Iter-BR}(\actvec_\S, \actvec_{\C_0}, a)$. 
\EndFor
\For{Each selected player $a$}
\State Update $\actvec_\S$ by setting $\alpha'_a$ instead of $\alpha_a$.
\EndFor
\State $t\gets t+1$.
\State Calculate $\epsilon_2$ for $\alpha_\S$.
\State Update $\epsilon_2$ if lower than the current value.
\EndWhile
\State Compute $\Co_g(\actvec)$.
\State If lower than the current value update $\Co_g$ and set $\actvec = <\alpha'_g, \actvec_\S, \actvec_\C>$. 
\State $\alpha'_g\gets \alpha'_g+\Delta_g$.
\EndWhile
\State return $\actvec$.
\end{algorithmic}
\end{algorithm}
\fi

Note that QIP-Best-Response computes simultaneously strategies for a State under consideration and all best-responding Counties (under all States), fixing the strategies of the Government and other States (Line~\ref{qip_step} of Algorithm~\ref{alg:taylor_iter}). To summarize, we replace the each County-level quadratic program by a system of (mixed-integer) linear constraints, and use a concatenation of these systems for all Counties as the constraint system for the State-level optimization problem under consideration.   

However, the performance of QIP is sensitive to the point around which we approximate the infection cost with Taylor series. 
As our strategy space $[0,1]$ is bounded, we propose an iterative  algorithm called Taylor-Iter-BR (Algorithm~\ref{alg:taylor_iter}) which is inspired by trust-region methods in optimization~\cite{nocedal2006numerical}. The idea is that we approximate the infection cost around the action profile we found in the previous iteration. We compare Algorithm~\ref{alg:taylor_iter} with Algorithm~\ref{alg:ne} for run-time and solution quality in the compliant setting (Section~\ref{sec:eqbm_select} Remark~\ref{compli_setting}) under a symmetric transport matrix, i.e., $r_{aa'}=\frac{1}{|\C|}=\frac{1}{10}$ for any two Counties $a,a'$; the goal is to compute the best response for the Government.
As shown in the Figure~\ref{fig:QIP}, Taylor-Iter-BR with only 2 iterations works as well as BRD with a grid size of $0.01$, but it runs significantly faster. We report further experiments comparing the two approaches in Appendix~\ref{sec:methods-qip-appx}.

\begin{figure}[ht!]
    \centering
	\begin{subfigure}[t]{0.5\textwidth}
		\includegraphics[scale=0.4]{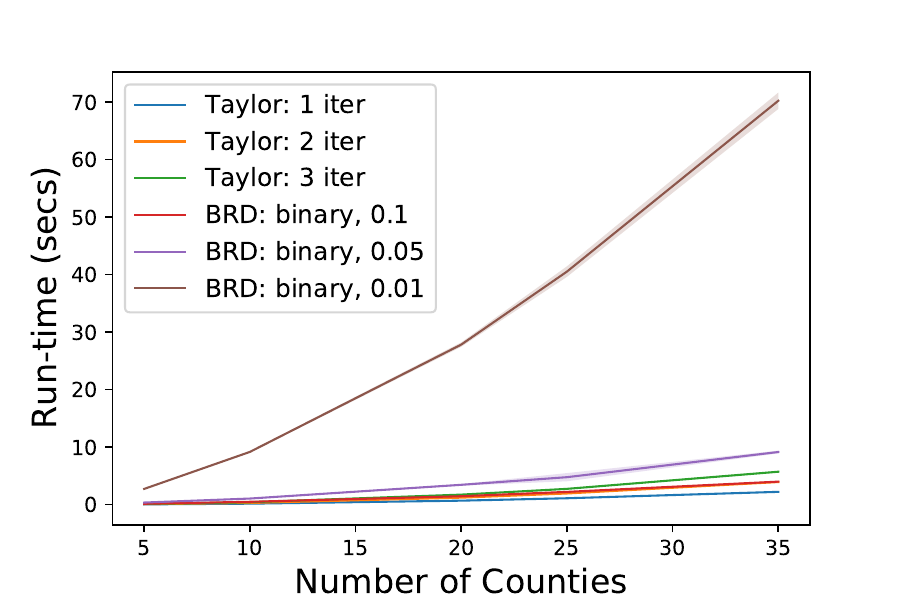}
		\caption{Run-time of BRD and Taylor-Iter-BR. \label{fig:run_time}}
	\end{subfigure}~
	\begin{subfigure}[t]{0.5\textwidth}
		\includegraphics[scale=0.36]{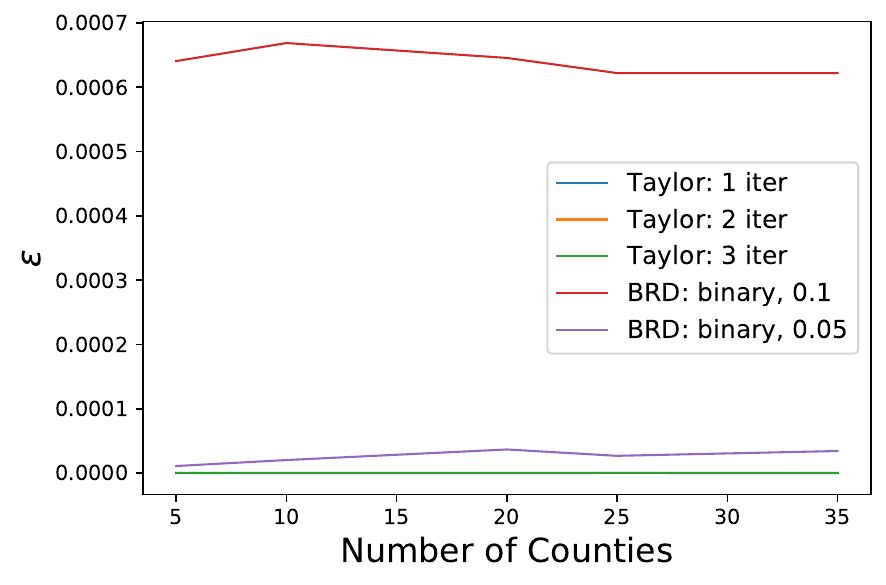}
		\caption{Performance of BRD and Taylor-Iter-BR. \label{fig:1_N_ep}}
	\end{subfigure}\\
	\begin{subfigure}[t]{0.5\textwidth}
		\includegraphics[scale=0.37]{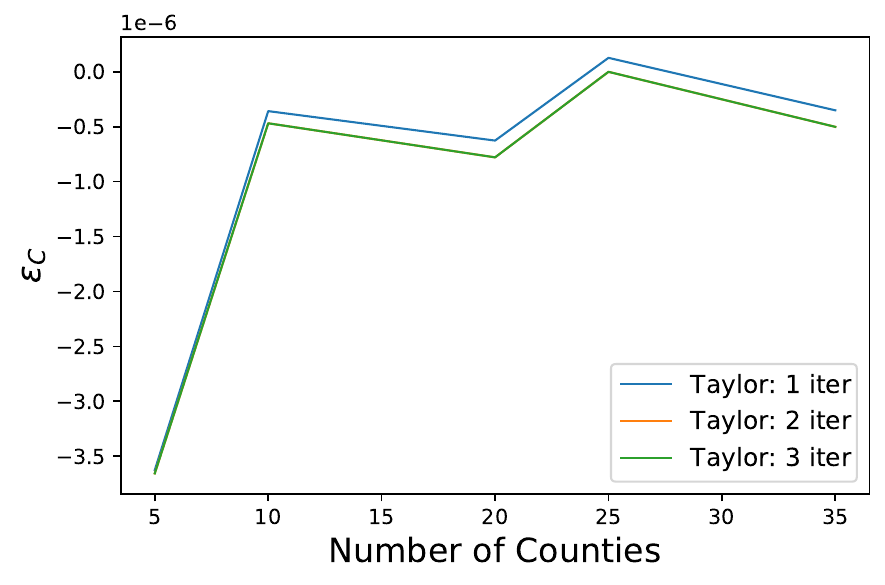}
		\caption{$\epsilon_\C$ of Taylor-Iter-BR.  \label{fig:1_N_ep_M}}
	\end{subfigure}
	\caption{Run-time and performance comparison of BRD with grid size $\Delta_g = 0.05,0.1$ and Taylor-Iter-BR with 1 to 3 iterations as a function of number of Counties. Performance $\epsilon$ is evaluated by the BRD with grid size $\Delta_g = 0.01$. }
	\label{fig:QIP}
\end{figure}

Using Algorithm~\ref{alg:taylor_iter} as a subroutine, we design Algorithm~\ref{alg:qip-brd} as follows for approximate SPNE computation. The Government performs a grid search over all of its possible actions. For each of these actions, we can run the best response dynamics over the States (until convergence), knowing that each State can compute its best response using Algorithm~\ref{alg:taylor_iter} which also determines the action of the Counties. The Government then selects an action that maximizes its payoff which in turn determine the actions of the States and Counties. The advantage of this approach is that search is not necessary at the County-level.

\section{Experiments}\label{sec:expts}
We will now report our experiments that demonstrate how the HECG model (Section~\ref{sec:game_model}) and equilibrium approximation methodologies (Section~\ref{sec:methods}) can offer insights into phenomena emerging from structured strategic interactions among hierarchical policy-makers. We set up two different \textit{worlds} based on synthetic data (Section~\ref{sec:expt_synth}) and real-world data (Section~\ref{sec:ny-nj}). In each world, we conducted three sets of experiments:
\begin{itemize}
\item In the first set (Sections~\ref{sec:NYNJ-exp-soc-wel} and~\ref{sec:exp-soc-wel}), we gauged how different parameter configurations, including different degrees of centralization, impact the overall social welfare of all agents in our HECG.
\item In the second set (Sections~\ref{sec:exp-free-riding} and~\ref{sec:NYNJ-free-riding}), we explored conditions under which evidence of free-riding is or is not observed in equilibrium.
\item In the final set (Sections~\ref{sec:exp-fairness} and \ref{sec:NYNJ-fairness}), we investigated how different degrees of centralization and mismatched priorities of HECG agents impact fairness in the distribution of costs.
\end{itemize}
We explain how we quantify the (negative) social welfare, proxy for free-riding, and the cost distribution fairness measure, applicable to both worlds, in Sections~\ref{sec:NYNJ-free-riding}, ~\ref{sec:exp-free-riding}, and~\ref{sec:exp-fairness} respectively. 
Given the large and complex set of parameters characterizing our multi-agent system, we report our salient experimental results only. 

\subsection{Experiments On Synthetic Data}\label{sec:expt_synth}
In all these experiments, we have two States, denoted simply by $1$ and $2$ (also in subscripts), and five Counties in each of the States  with equal populations (i.e., shares). Thus, $n=10$ and $n_1=5$ (see Figure~\ref{fig:model}). We will denote the policies of the two States by $\alpha_1$ and $\alpha_2$ respectively. 
For approximate equilibrium computation, we used BRD (Section~\ref{sec:brd} Algorithm~\ref{alg:ne}). 

\subsubsection{Social welfare}\label{sec:exp-soc-wel}
In the first set of experiments, we measure the negative social welfare --- or \textit{social cost} --- as the overall cost of infection and implementation from the Government's perspective as defined in Section~\ref{sec:model}: $\Co_{g}(\actvec) = \kappa_{g} \Cinc_{g}(\actvec) + (1-\kappa_{g})\Cdec_{g}(\actvec)$.
We compare this cost across four different (de)centralization scenarios and different parameter configurations for each scenario (see below). The centralization scenarios are as follows:
\begin{itemize}
    \item Centralized County-Specific (\CCS): The Government selects a potentially different policy for each of the $n$ Counties to minimize its own overall cost $\Co_{g}(\actvec)$. This entails solving an optimization problem in the space $[0,1]^n$. Here, States are immaterial and Counties simply follow the respective policies set by the Government.
    \item Centralized Uniform (\CU): The Government selects a single "nationwide" policy to minimize its own overall cost, and each County follows this policy. This can be viewed as a version of the \CCS optimization problem where all County-level policies are constrained to be equal, hence the solution space is $[0,1]$. 
    \item 2-level Equilibrium (\EQtwo): Only the Government and States are autonomous; each County is constrained to comply with its respective State. This is identical to the compliant setting introduced in Section~\ref{sec:eqbm_select} Remark~\ref{compli_setting}. We compute the solution by the 2-level version of the BRD algorithm (Section~\ref{sec:brd} Algorithm~\ref{alg:ne}, but taking $l=2$ as the lowest level).
    \item 3-level Equilibrium (\EQthree): The Government, all States, and all Counties are autonomous. This corresponds to the (full) HECG as delineated in Section~\ref{sec:model}. We compute the solution by the (full) BRD algorithm as described in Section~\ref{sec:brd} Algorithm~\ref{alg:ne}.
\end{itemize}
Our parameter configurations are as follows. We perform experiments for all the above scenarios under \textit{transport symmetry}, which means that we set the transport matrix $R$ proportional to an identity matrix: $r_{aa'}=\frac{1}{|\C|}=\frac{1}{10}$ for any two Counties $a,a'$.  The two States have equal populations $N_1=N_2=500$, hence $\mu_1=\mu_2=0.5$. States $a \in \{1,2\}$ have identical weight vectors where we vary the shared non-compliance cost weight $\gamma$ in $[0,1]$ and set $\kappa_a = 0.9(1-\gamma)$ and $\eta_a = 0.1(1-\gamma)$ accordingly; we also vary the Government's single weight $\kappa_g$ in $(0,1)$ (see below for details). 

For every scenario, all Counties in the same State have equal initial infection rates. We fix the initial infection rate of every County in State 1 at $0.1$ (low) while that of every County in State 2 is a parameter that we vary in $\{0.7,0.8,0.9\}$; we refer to the latter as \II.
Under \CCS and \CU, there is no possibility of non-compliance, by design. Under \EQtwo, the States' weight vectors as described above complete the instantiation of all parameters. For \EQthree, each County has a weight vector identical to that of the respective State (i.e., all agents but the Government have the same weight vector). Thus, \CU/\CCS, \EQtwo, and \EQthree represent increasing levels of decentralization, i.e., discretion given to lower-level agents in the hierarchy; within each of the \EQtwo and \EQthree scenarios, a lower value of the shared non-compliance cost weight $\gamma$ corresponds to a higher degree of decentralization, and the special case $\gamma=1$ coincides with \CU. 

Note that, for any parameter configuration, the \CCS optimal social cost is the best achievable social cost across all scenarios considered and acts a lower bound for those under \CU, \EQtwo, and \EQthree by design. However, it is not \textit{a priori} obvious whether \EQtwo and \EQthree, where lower-level policy-makers have some autonomy, can improve upon \CU, the constrained but computationally and conceptually simpler centralized solution, in terms of the social cost. 

For \CCS and \CU, we use projected stochastic gradient descent to compute the optimal social cost.\footnote{We used a learning rate of $0.2$ and ran $10^4$ iterations of stochastic gradient descent (SGD) in PyTorch, projecting each iterate back to $[0,1]$ for \CU and $[0,1]^n$ for \CCS. For each value of \II, we ran SGD thrice, initialized at $0, 0.5, 1$ respectively for \CU and at vectors of length $n$ with each entry equal to $0, 0.5, 1$ respectively for \CCS; finally, we took the minimum social cost over the 3 runs as our solution for the corresponding scenario.} 
For \EQtwo and \EQthree, the BRD discretization factor for best-response search is $0.05$ and approximation error tolerances are set as $e_1=e_2=e_3=10^{-6}$; the results are averaged over $10$ trials, each trial characterized by a random initialization of BRD, and the error bars (one standard error) are shown as shaded regions (for most curves, they are too small to see at the scale of our plots). We present a representative sample of our results in Figures~\ref{fig:soc_cost} and~\ref{fig:soc_cost_gamma}.
\begin{figure}[ht!]
    \centering
	\begin{subfigure}[t]{0.48\textwidth}
		\includegraphics[scale=0.38]{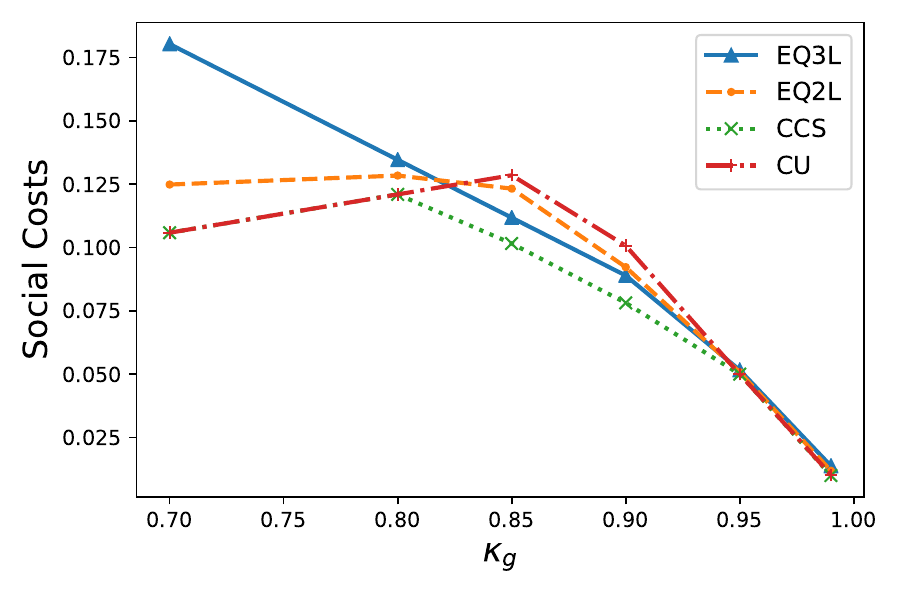}
		\caption{\II $=0.8$, $\gamma=0.15$. \label{fig:soc_cost1}}
	\end{subfigure}~
	\centering
	\begin{subfigure}[t]{0.48\textwidth}
		\includegraphics[scale=0.38]{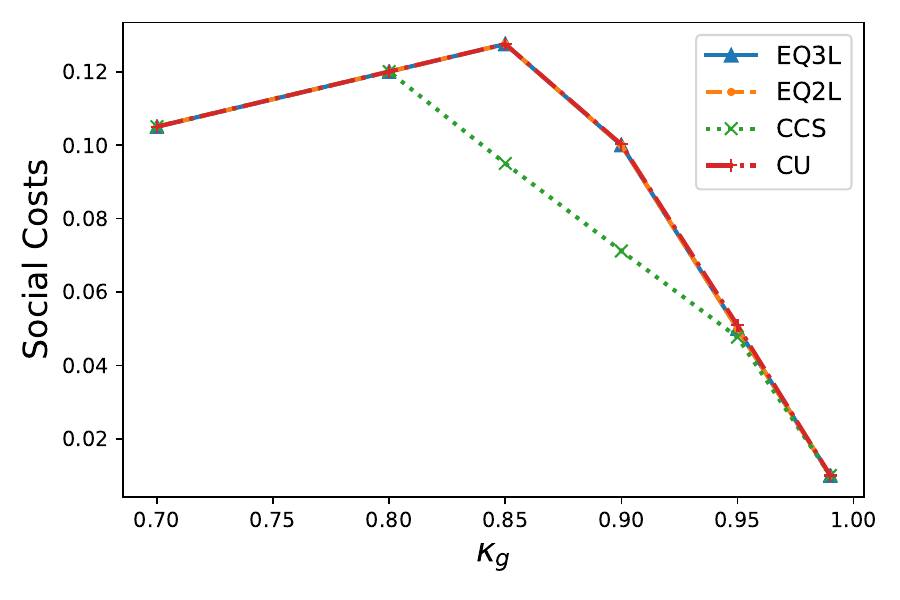}
		\caption{\II $=0.9$, $\gamma=0.8$. \label{fig:soc_cost2}}
	\end{subfigure}
	\caption{Social cost (y-axis) for each of our centralization scenarios as a function of the Government's infection cost weight (x-axis), with \II and $\gamma$ as parameters (two pairs of parameter-values shown). Note that $\gamma$ is irrelevant for \CCS and \CU.}
	\label{fig:soc_cost}
\end{figure}

Our first observation is that the social cost exhibits a non-monotonic (in general), concave variation with respect to $\kappa_g$ (see Figure~\ref{fig:soc_cost}). Moreover, consistently for all values of \II, social costs for \CCS and \CU coincide except possibly over a small regime $\kappa_g \in [0.80,0.95]$. Thus, outside of this regime, decentralization (\EQtwo and \EQthree) cannot offer any improvement in social cost over the centralized solution \CU. Figures~\ref{fig:soc_cost_gamma1} and~\ref{fig:soc_cost_gamma2} depict the variation in social cost with $\gamma$ under \EQtwo and \EQthree and their comparison with that under \CU/\CCS for two representative values of $\kappa_g$ outside the above regime. We see that, in general, lower values of $\gamma$ result in higher social costs, with \EQthree being at least as high as \EQtwo; but the costs are identical to that for centralized solutions above a high enough value of $\gamma$ (which could be as low as $0.4$ depending on the parameters). This provides evidence that decentralization can be detrimental to social cost unless $\kappa_g$ is high (but not too close to $1$).

\begin{figure}[ht!]
    \centering
	\begin{subfigure}[t]{0.45\textwidth}
		\includegraphics[scale=0.37]{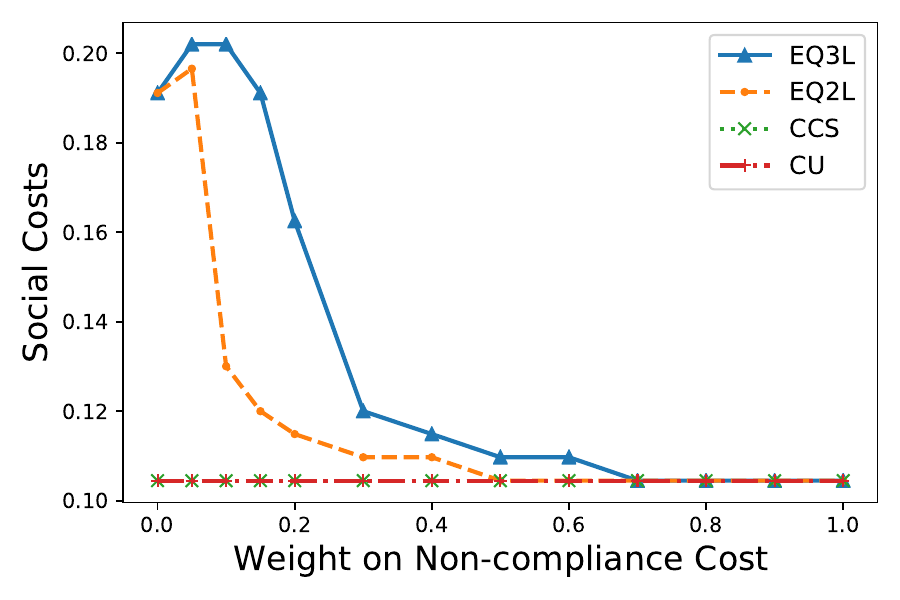}
		\caption{\II $=0.7$, $\kappa_g=0.7$ \label{fig:soc_cost_gamma1}}
	\end{subfigure}~
	\begin{subfigure}[t]{0.45\textwidth}
		\includegraphics[scale=0.37]{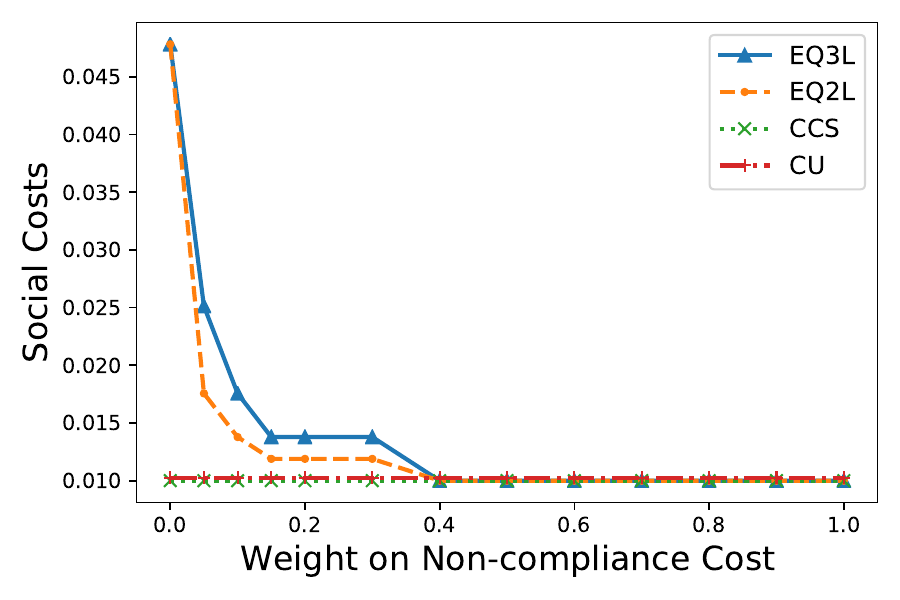}
		\caption{\II $=0.8$, $\kappa_g=0.99$ \label{fig:soc_cost_gamma2}}
	\end{subfigure}\\ 
	\begin{subfigure}[t]{0.45\textwidth}
		\includegraphics[scale=0.37]{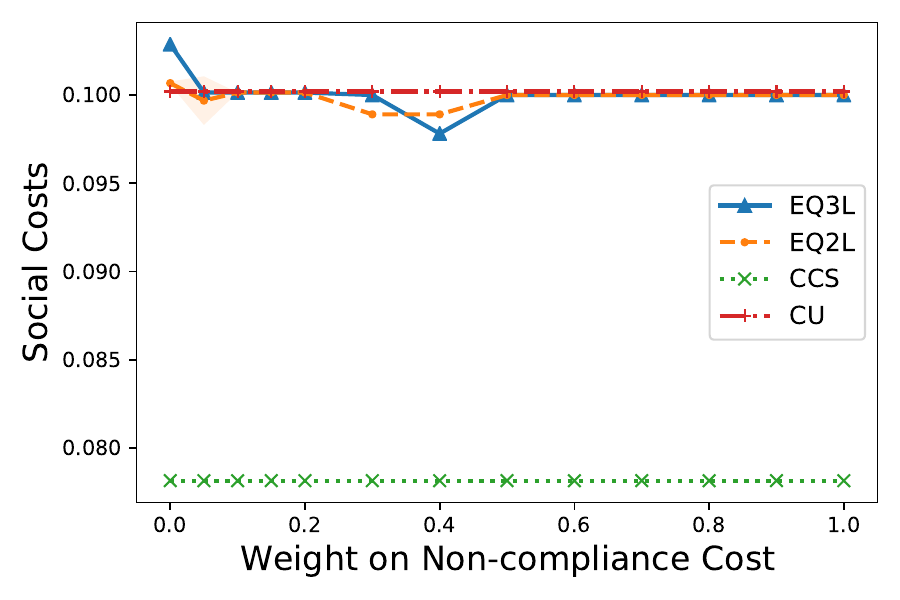}
		\caption{\II $=0.7$, $\kappa_g=0.9$ \label{fig:soc_cost_gamma3}}
	\end{subfigure}~
	\begin{subfigure}[t]{0.45\textwidth}
		\includegraphics[scale=0.37]{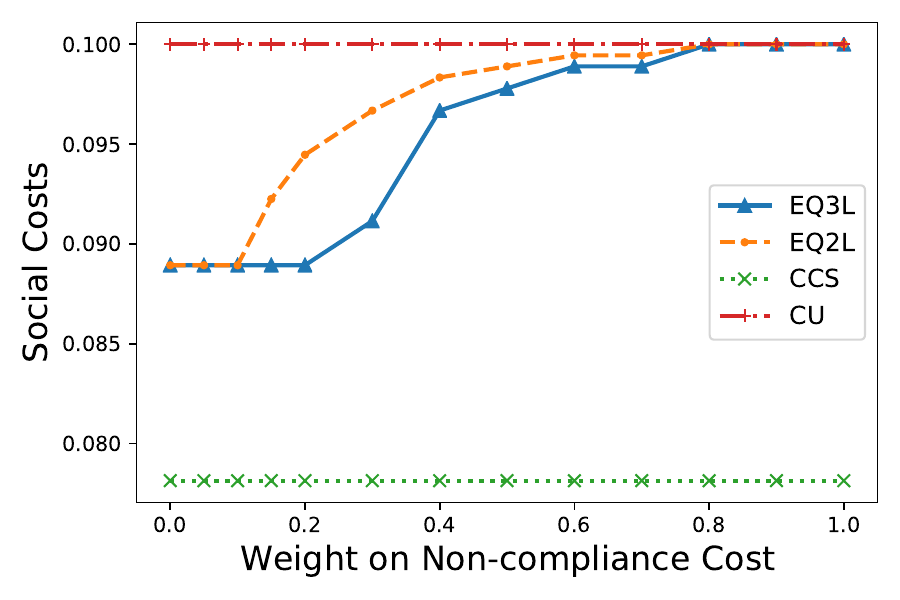}
		\caption{\II $=0.8$, $\kappa_g=0.9$ \label{fig:soc_cost_gamma4}}
	\end{subfigure}\\
	\begin{subfigure}[t]{0.45\textwidth}
		\includegraphics[scale=0.37]{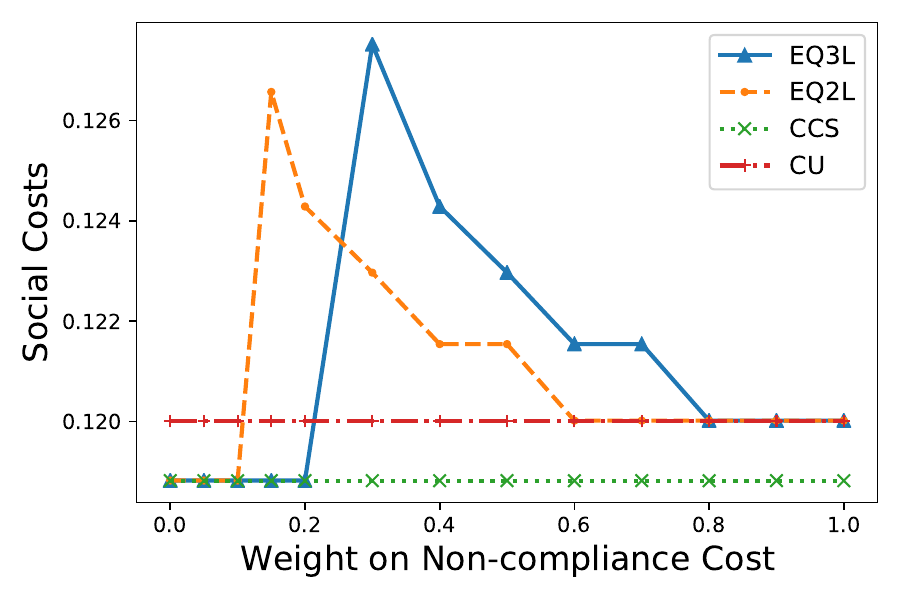}
		\caption{\II $=0.9$, $\kappa_g=0.8$ \label{fig:soc_cost_gamma5}}
	\end{subfigure}~
	\begin{subfigure}[t]{0.45\textwidth}
		\includegraphics[scale=0.37]{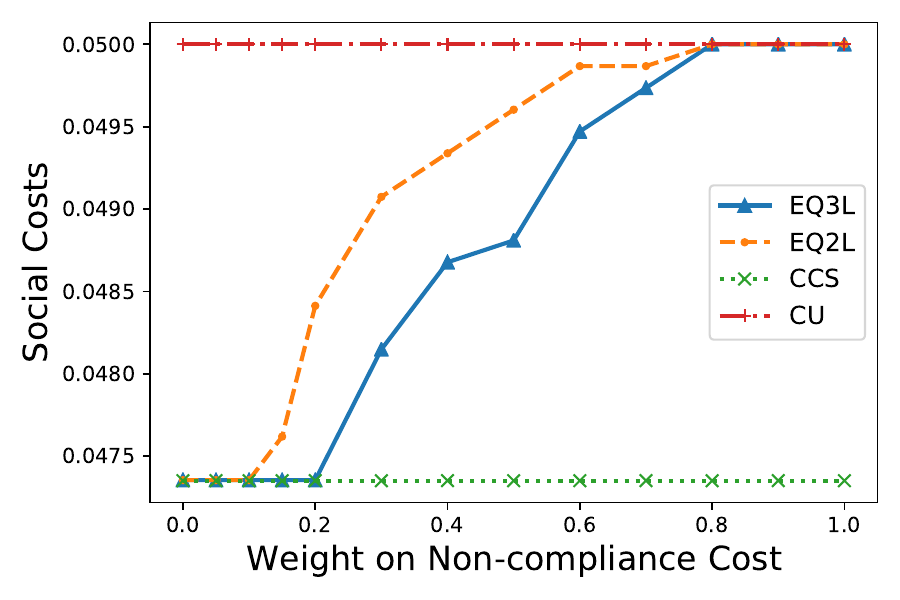}
		\caption{\II $=0.9$, $\kappa_g=0.95$ \label{fig:soc_cost_gamma6}}
	\end{subfigure}
	\caption{Social cost (y-axis) for each of our centralization scenarios as a function of the shared non-compliance cost weight (x-axis), with \II and $\kappa_g$ as parameters.}
	\label{fig:soc_cost_gamma}
\end{figure}

Let us now turn to the more interesting regime $\kappa_g \in [0.80,0.95]$ (Figures~\ref{fig:soc_cost_gamma3} through~\ref{fig:soc_cost_gamma6}). There exist parameter configurations (Figures~\ref{fig:soc_cost_gamma3},~\ref{fig:soc_cost_gamma5}) for which the decentralized scenarios are comparable in social cost to \CU (notice the difference in scale between Figures~\ref{fig:soc_cost_gamma3} and~\ref{fig:soc_cost_gamma5}, and the smaller gap between \CU and \CCS in the latter). But more importantly, there are values of \II for which lower $\gamma$ (willingness to comply) generally leads to lower social costs than \CU, \EQthree being weakly lower than \EQtwo  (Figures~\ref{fig:soc_cost_gamma4},~\ref{fig:soc_cost_gamma6}), sometimes even as low as \CCS (Figure~\ref{fig:soc_cost_gamma6}). Thus, the main takeaway from these experiments is the identification of a non-trivial set of conditions under which decentralization can offer better social costs than a uniform centralized policy.

\subsubsection{Free-riding}\label{sec:exp-free-riding}
In this section, we will characterize free-riding between States. We begin with the rationale for our measurement and visualization of free-riding. Consider the compliant setting (\EQtwo). Suppose State 2 has a higher initial infection rate than State 1; then, intuitively, it may prefer a weak distancing policy ($\alpha_{2} \gg 0$) to reduce its implementation cost as most of its population is already infected; at the same time, since infection can spread from State 2, State 1 could still suffer a large infection cost unless it employs a strong distancing policy ($\alpha_{2} \ll 1$). This creates the possibility for State 2 to free-ride off State 1; but, whether this actually happens depends on the combination of parameters (including non-compliance issues with the Government), and the same possibility may be created by different sets of conditions. Based on this intuition, we use the difference between policy strengths of the States as an indicator of the degree of free-riding. Intuitively, as $\alpha_2-\alpha_1$ approaches $1$ (resp. $-1$), the degree of free-riding of State 2 (resp. State 1) off State 1 (resp. State 2) increases. How should we operationalize free-riding for non-compliant settings, i.e., when Counties have discretion too? For such settings, $\alpha_a$ is merely a recommendation by State $a$, and the County-level equilibrium policies are the ones that are realized (potentially different across Counties under the same State). Hence, in our free-riding proxy, we replace $\alpha_a$ with $\langle\alpha_a \rangle$ defined as the average of the (equilibrium) policies set by all Counties in State $a \in \{1,2\}$. Clearly, $\langle\alpha_a \rangle=\alpha_a$ for the compliant setting; hence, $\langle\alpha_2 \rangle-\langle\alpha_1 \rangle$ is our general proxy for the degree of free-riding by State 2.

\begin{figure}[ht!]
    \centering
	\begin{subfigure}[t]{0.45\textwidth}
 		\includegraphics[scale=0.36]{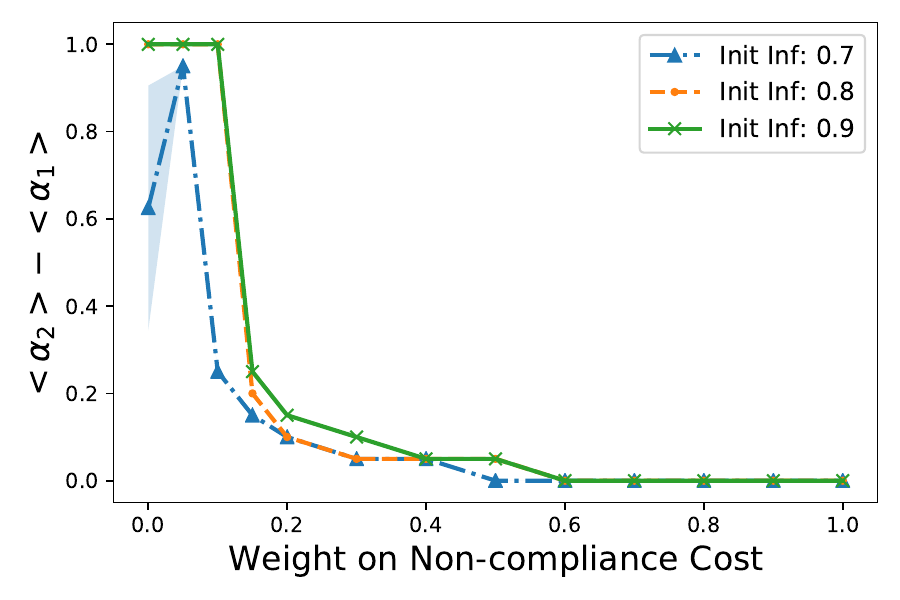}
		\caption{\EQtwo, $\kappa_g=0.7$ \label{fig:free_comp1}}
	\end{subfigure}~	\begin{subfigure}[t]{0.45\textwidth}
		\includegraphics[scale=0.36]{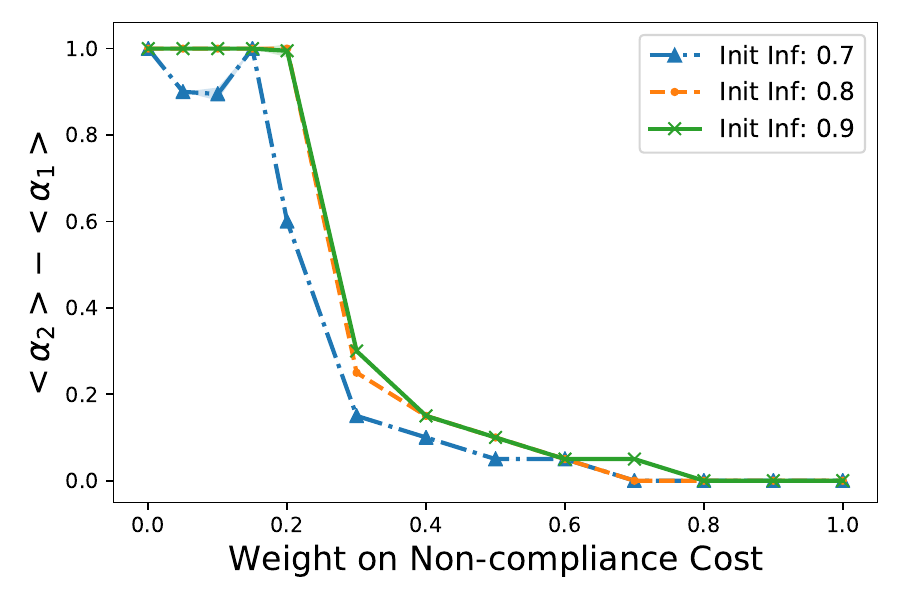}
		\caption{\EQthree, $\kappa_g=0.7$ \label{fig:freeride_noncompcounties1}}
	\end{subfigure}\\
	\begin{subfigure}[t]{0.45\textwidth}
 		\includegraphics[scale=0.36]{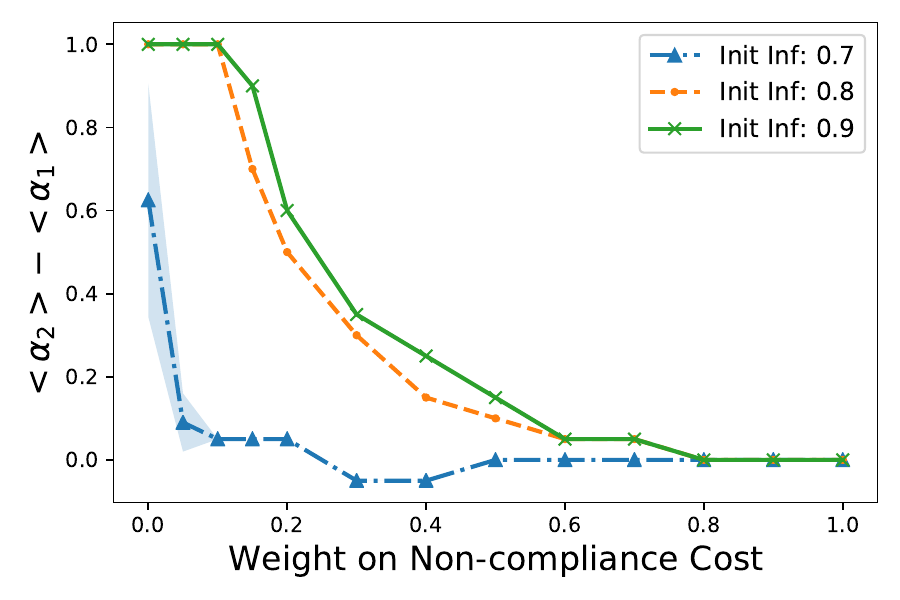}
		\caption{\EQtwo, $\kappa_g=0.9$ \label{fig:free_comp2}}
	\end{subfigure}~
		\begin{subfigure}[t]{0.45\textwidth}
		\includegraphics[scale=0.36]{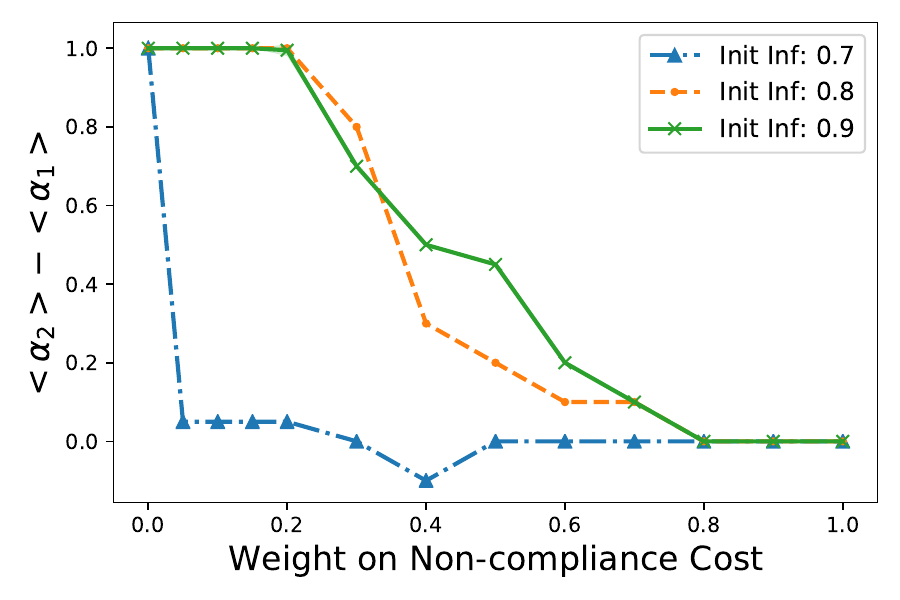}
		\caption{\EQthree, $\kappa_g=0.9$ \label{fig:freeride_noncompcounties2}}
	\end{subfigure}\\
	\begin{subfigure}[t]{0.45\textwidth}
 		\includegraphics[scale=0.36]{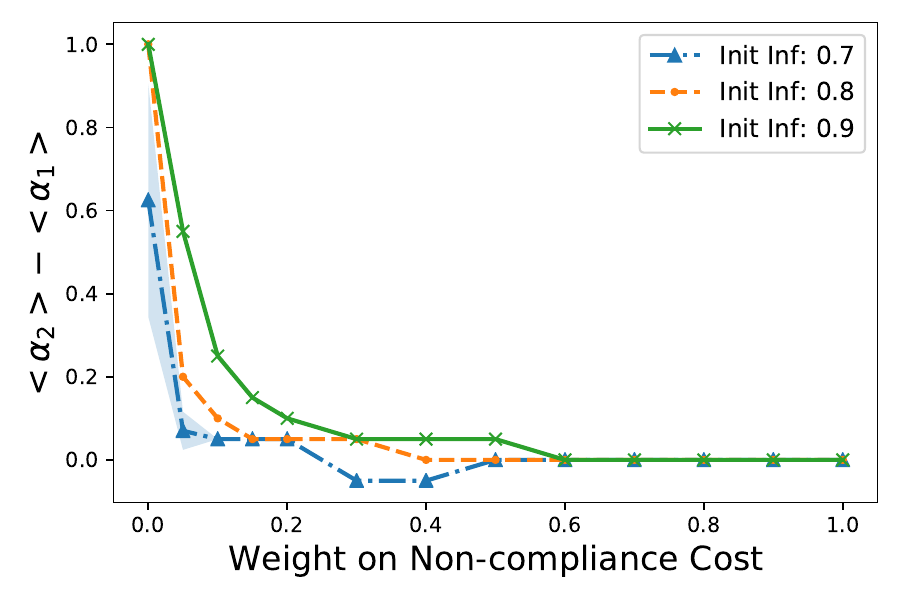}
		\caption{\EQtwo, $\kappa_g=0.99$ \label{fig:free_comp3}}
	\end{subfigure}~
	\begin{subfigure}[t]{0.45\textwidth}
		\includegraphics[scale=0.36]{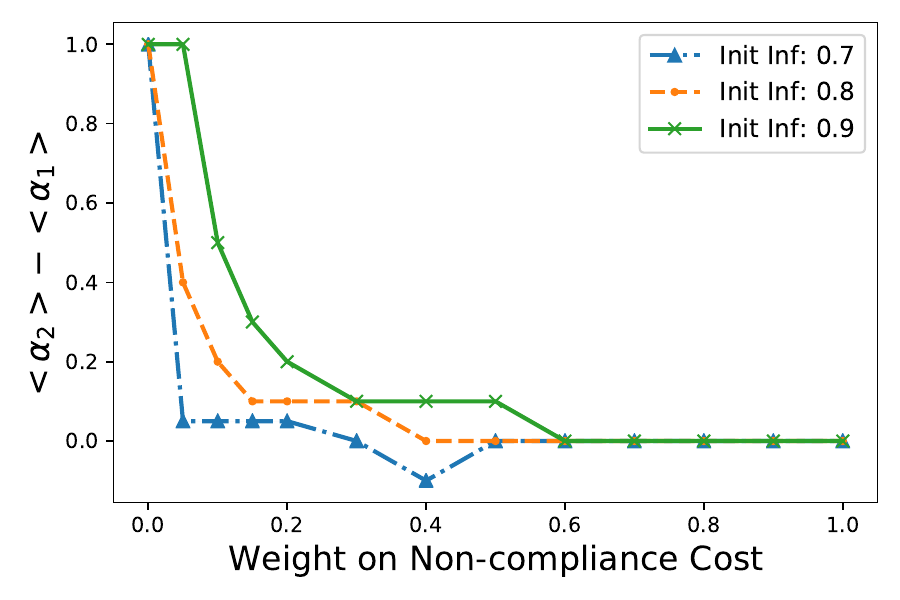}
		\caption{\EQthree, $\kappa_g=0.99$ \label{fig:freeride_noncompcounties3}}
	\end{subfigure}
	\caption{Free-riding (y-axis) as a function of non-compliance cost weight (x-axis) under transport symmetry. Each curve corresponds to a different initial infection rate of State 2 (Init Inf) as indicated in the legend. The left (resp. right) column corresponds to the compliant (resp. non-compliant) scenario.}
	\label{fig:freeriding}
\end{figure}

Our first subset of reported experiments are performed on settings with transport symmetry, as described in Section~\ref{sec:exp-soc-wel}. Other parameter configurations are analogous to those of the \EQtwo (compliant) and \EQthree (non-compliant) scenarios of Section~\ref{sec:exp-soc-wel}. As before, the crucial difference between States is in the initial infection rates only.
Figure~\ref{fig:freeriding} depicts the variation in the States' averaged County-level policy-difference $\in [-1,1]$ against their shared non-compliance weight $\gamma \in [0,1]$ for three representative values of the Government's infection cost weight $\kappa_g$ and each of the two scenarios with decentralization.
The results are averaged over $10$ trials, each trial characterized by a random initialization of BRD, and the error bars (one standard error) are shown as shaded regions. 

The broad findings for both the compliant (Figures~\ref{fig:free_comp1},~\ref{fig:free_comp2} and~\ref{fig:free_comp3}) and non-compliant (Figures~\ref{fig:freeride_noncompcounties1},~\ref{fig:freeride_noncompcounties2}, and~\ref{fig:freeride_noncompcounties3}) scenarios are similar. We observe that, when the States' shared non-compliance cost weight is close to $0$, 
there is strong evidence of free-riding by the high-infection State. This is understandable since under these conditions, the Government's recommendation is immaterial, and States act selfish-rationally and (almost) independently, consistent with the intuition presented in the first paragraph of Section~\ref{sec:exp-free-riding}. However, if States care even a little about non-compliance, the policy-difference drops  towards $0$ as we increase the non-compliance weight. The value of the non-compliance cost weight at which the drop-off begins (resp. the rate of drop-off) shows an increasing (resp. a decreasing) trend as \II increases from $0.7$ to $0.8$ to $0.9$. 
For very high values of the non-compliance weight cost, the policy-difference vanishes as all agents comply with the Government's recommendation. This suggests that decentralization may result in free-riding but only for a small regime of low compliance (smaller values of $\gamma$) and a large difference between initial infection rates between the States ($0.1$ vs $0.8$ or higher). However, for a high enough $\kappa_g$ ($0.99$ in our experiments), there is practically no evidence of free-riding as long as there is even a little inclination to comply.

We now present another subset of experiments that we performed in a compliant setting (analogous to \EQtwo in the above subset) but under a particular type of \textit{transport asymmetry}. There are many ways in which $R$ could be asymmetric; we focus on a particular class of asymmetries where one subset of Counties $\C_f\subset\C$ are \textit{globally favorite destinations} (and equally popular) and all others are equally (un)popular --- quantitatively, for each County $j$, $r_{ij}=r^H$ for each $i \in \C_f$, and  $r_{kj} = r^L$ for each $k \in \C\setminus\C_f$ for some $0 < r^L < r^H < 1$, and $\sum_{i \in \C} r_{ij}=1$. More specifically, for the remaining experiments on free-riding, we make all Counties in State 1 our globally favorite destinations (i.e., $\C_f = \{c \in \C \mid \pi(c)=1\}$), with $r^H = 0.16$, hence $r^L = 0.04$.

Another difference from the previous subset of experiments is that we try three different values $\{0.1, 0.5, 0.9\}$ for State 1's (per-County) initial infection rate, and vary State 2's initial infection rate (still referred to as \II) as before in $\{0.1, 0.2, \dots, 0.9\}$ for each of the above three choices. This is based on the intuition that, owing to transport asymmetry, the symmetry between States in terms of the initial infection rates is also lost. 
The remaining parameter configurations are similar to those of the \EQtwo scenario in the first subset -- in particular, States have the same non-compliance cost weight $\gamma$ and infection cost weight $0.9(1-\gamma)$, with $\gamma$ and $\kappa_g$ being varied as parameters. 
Figure~\ref{fig:freeride_asym} shows a representative sample of our results ($\langle\alpha_2 \rangle-\langle\alpha_1 \rangle$ averaged over $10$ random BRD initializations; one standard error shown as a shaded region which is, again, too small to see at this scale). 

\begin{figure}[ht!]
\centering
	\begin{subfigure}[t]{0.33\textwidth}
 		\includegraphics[scale=0.27]{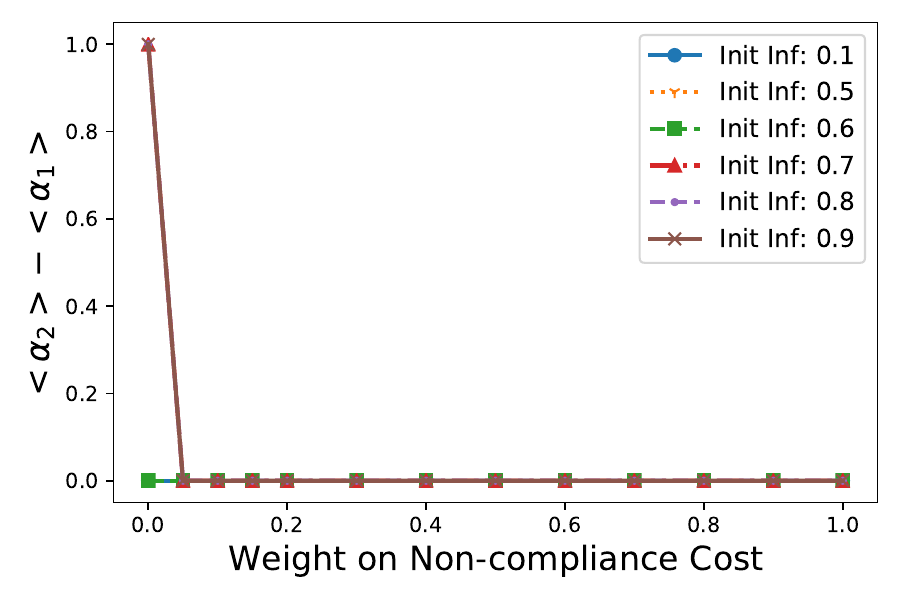}
		\caption{$\kappa_g=0.7$, State 1: $0.1$ \label{fig:asym_kg_p7_1}}
	\end{subfigure}~
	\begin{subfigure}[t]{0.33\textwidth}
 		\includegraphics[scale=0.27]{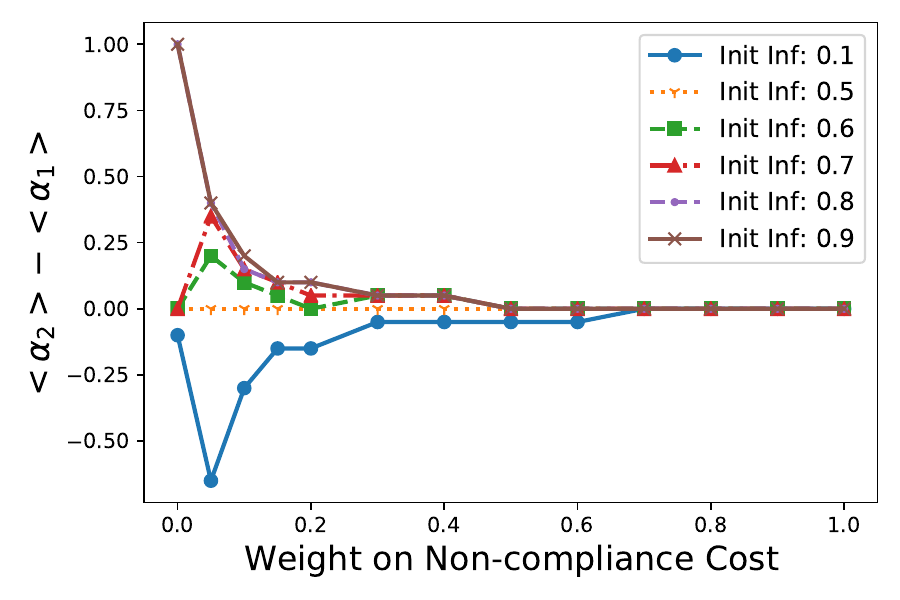}
		\caption{$\kappa_g=0.7$, State 1: $0.5$ \label{fig:asym_kg_p7_2}}
	\end{subfigure}~
	\begin{subfigure}[t]{0.33\textwidth}
 		\includegraphics[scale=0.27]{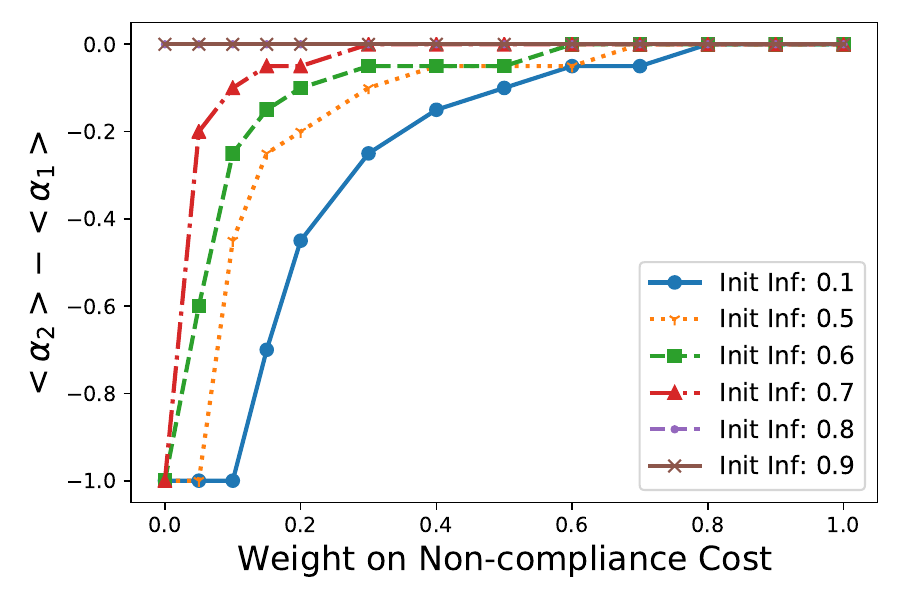}
		\caption{$\kappa_g=0.7$, State 1:  $0.9$ \label{fig:asym_kg_p7_3}}
	\end{subfigure}\\
	\begin{subfigure}[t]{0.33\textwidth}
 		\includegraphics[scale=0.27]{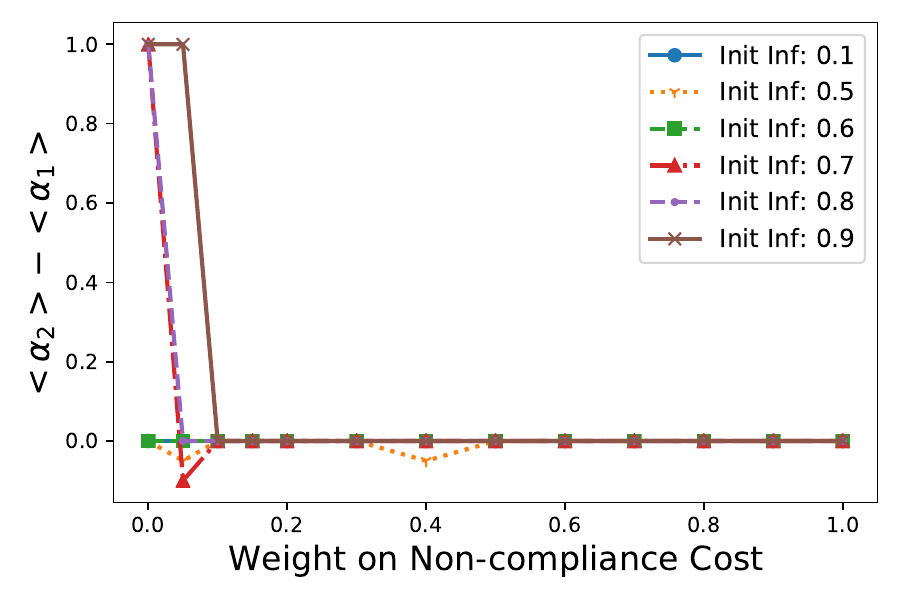}
		\caption{$\kappa_g=0.9$, State 1: $0.1$ \label{fig:asym_kg_p9_1}}
	\end{subfigure}~
	\begin{subfigure}[t]{0.33\textwidth}
 		\includegraphics[scale=0.27]{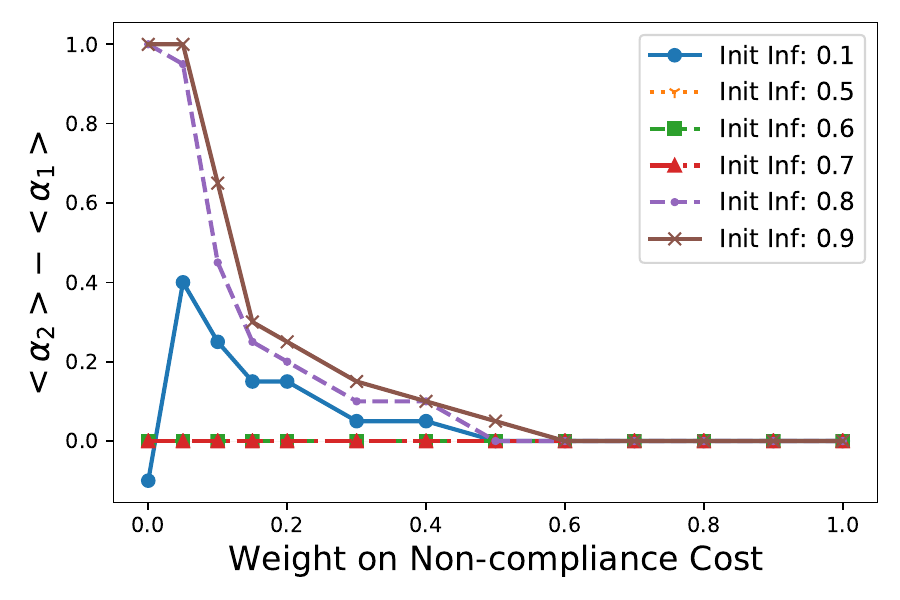}
		\caption{$\kappa_g=0.9$, State 1: $0.5$ \label{fig:asym_kg_p9_2}}
	\end{subfigure}~
	\begin{subfigure}[t]{0.33\textwidth}
 		\includegraphics[scale=0.27]{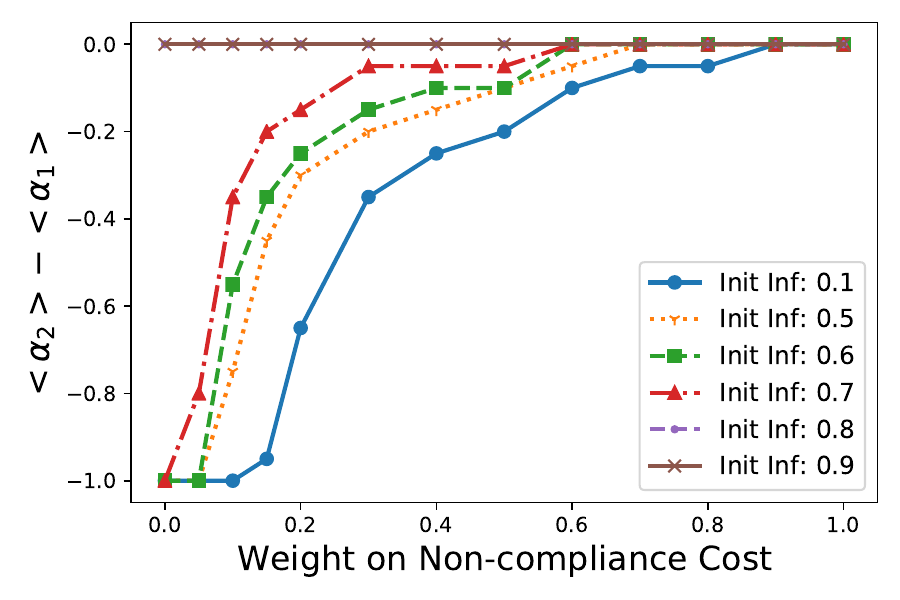}
		\caption{$\kappa_g=0.9$, State 1:  $0.9$ \label{fig:asym_kg_p9_3}}
	\end{subfigure}\\
	\begin{subfigure}[t]{0.33\textwidth}
 		\includegraphics[scale=0.27]{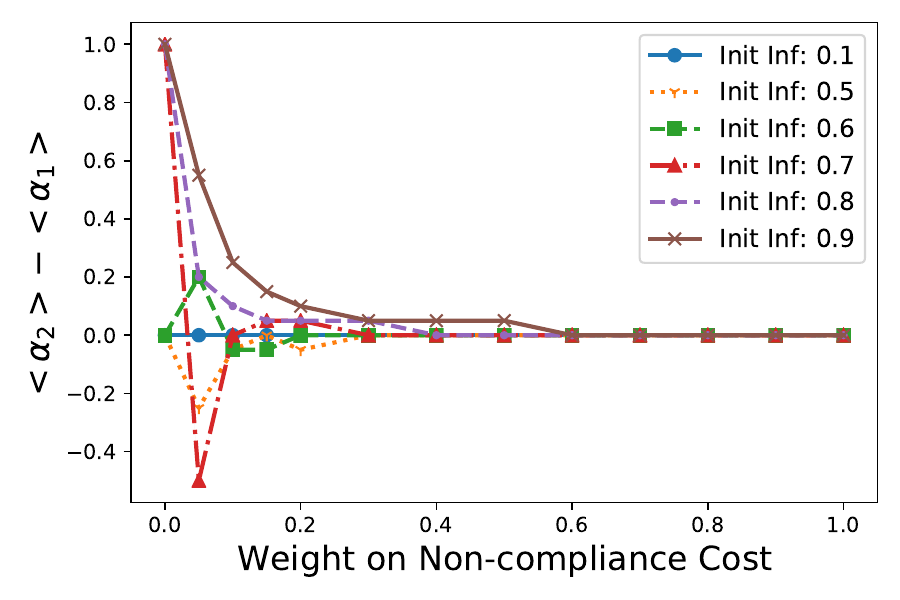}
		\caption{$\kappa_g=0.99$, State 1: $0.1$ \label{fig:asym_kg_p99_1}}
	\end{subfigure}~
	\begin{subfigure}[t]{0.33\textwidth}
 		\includegraphics[scale=0.27]{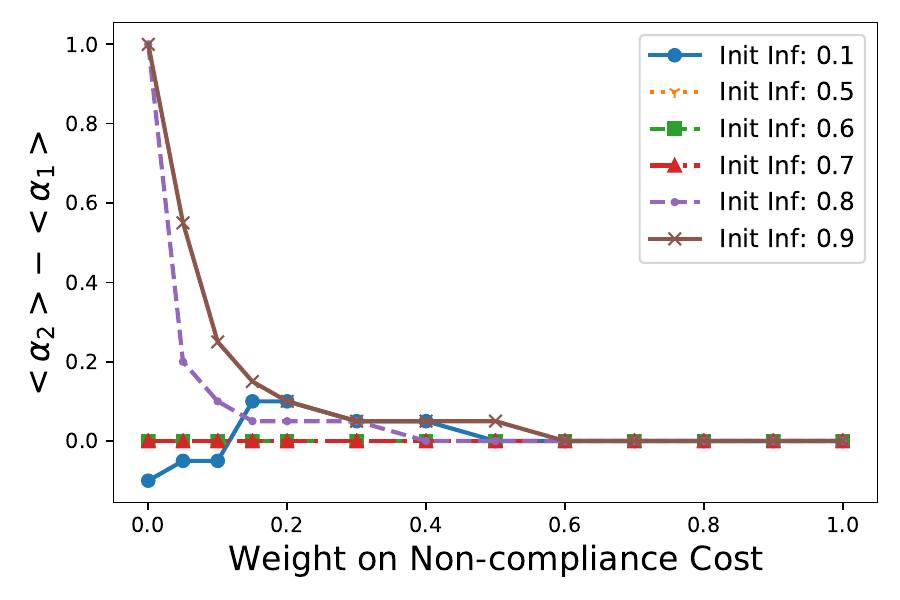}
		\caption{$\kappa_g=0.99$, State 1: $0.5$ \label{fig:asym_kg_p99_2}}
	\end{subfigure}~
	\begin{subfigure}[t]{0.33\textwidth}
 		\includegraphics[scale=0.27]{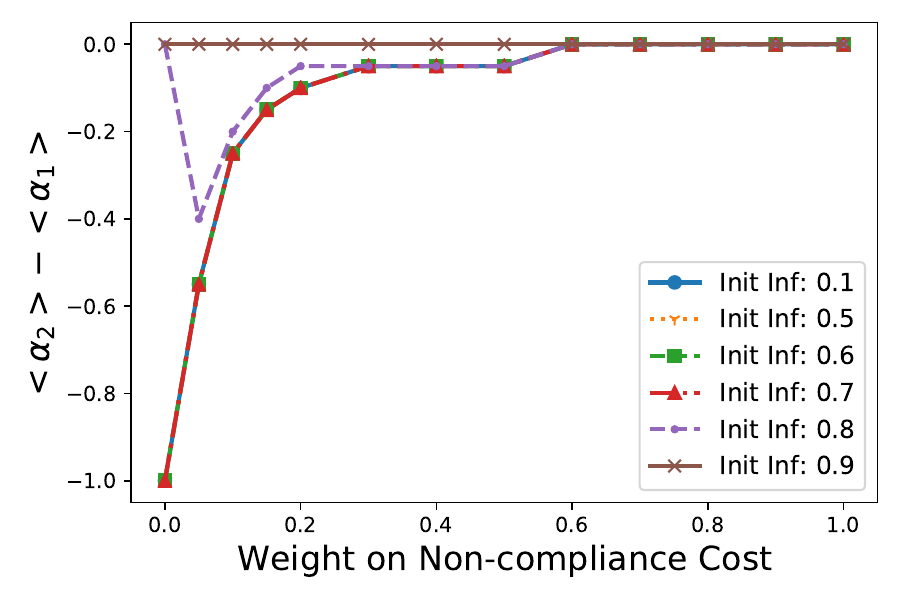}
		\caption{$\kappa_g=0.99$, State 1:  $0.9$ \label{fig:asym_kg_p99_3}}
	\end{subfigure}
	\caption{Free-riding (y-axis) as a function of non-compliance weight (x-axis) under transport asymmetry, Counties constrained to comply with respective States. Each row corresponds to a representative value of $\kappa_g$. The fraction after ``State 1'' is State 1's initial infection rate.}
	\label{fig:freeride_asym}
\end{figure}

The first column of panels (Figures~\ref{fig:asym_kg_p7_1},~\ref{fig:asym_kg_p9_1} corresponds to the lowest infection rate of State 1 in our experiments (i.e., an initial infection rate of 0.1). They show that there is no free-riding unless State $2$'s initial infection rate is significantly higher (at least $0.7$) \textit{and} the shared non-compliance weight is very low (roughly speaking, it has to be less than $0.2$ for a significant policy-difference to be observed). Even for the latter configurations, free-riding seems less and less prevalent as we decrease $\kappa_g$. For $\kappa_g=0.7$ (or smaller, not shown due to qualitative similarity), there is no free-riding except when Counties are fully autonomous (i.e., $\gamma=0$) and \II is at least $0.7$. Moreover, except perhaps for the outliers in Figure~\ref{fig:asym_kg_p99_1} for $\kappa_g=0.99$, it is State $2$ that free-rides off State $1$, if at all. This is understandable since State $1$ starts with a large uninfected population but has a higher exposure to infection because of including all the favorite destinations; hence, it makes sense for State $1$ to employ a stronger policy to prevent high infection costs, and State $2$ can take advantage of this. However, it is still interesting that even a small amount of centralization (i.e., a positive non-compliance weight cost) eliminates free-riding effectively. Finally, when $\kappa_g=0.99$ and State $1$ has the lowest initial infection rate of $0.1$, there seems to be mild evidence of free-riding by State $1$  off State $2$ for some values of \II (e.g., $0.7$).

The last column (Figures~\ref{fig:asym_kg_p7_3},~\ref{fig:asym_kg_p9_3}, and~\ref{fig:asym_kg_p99_3}) corresponds to the highest infection rate of State 1 in our experiments (i.e., an initial infection rate of 0.9). They stand in stark contrast to the previous setting. For these parameter configurations, it is State $1$ whose equilibrium policy is much weaker than that of State $2$. This is because the initial infection rate of State $2$ is lower than or equal to that of State $1$ --- if State $1$ takes no anti-infection measures, State $2$ is forced to adopt a strong policy to prevent new infections, despite its Counties being less favored destinations (in terms of the transport matrix). What is still interesting is that the magnitude of the policy-difference remains significantly higher than corresponding magnitudes in the first column (low initial infection rate in State $1$) as we increase the non-compliance cost weight. However, it is unclear whether this higher magnitude should be viewed as evidence of free-riding by State $1$ or State $1$ just giving up due to unfavorable initial conditions.

The middle column (Figures~\ref{fig:asym_kg_p7_2},~\ref{fig:asym_kg_p9_2}, and~\ref{fig:asym_kg_p99_2}) corresponds to an infection rate for State 1 that is in between the highest and lowest values we considered in the previous two settings (i.e., an initial infection rate of 0.5). In this case, State $2$'s initial infection rate may be higher than, equal to, or lower than that of State $1$. Despite this, if we find evidence of free-riding at all, it is almost always State $2$ that free-rides off State $1$. Moreover, the drop-off in the value of the policy-difference is slower in general than that in the first column, suggesting (slightly) higher evidence of free-riding of State $2$ off State $1$ even for higher values of the non-compliance cost weight (i.e., more and more centralized, hierarchical control). It thus appears that for a moderate initial infection rate in the globally favorite destinations, the transport asymmetry is a crucial factor in determining who free-rides and to what extent. Finally, there is evidence of free-riding in the opposite direction, i.e., by State $1$  off State $2$ (the magnitude of the policy-difference is $\approx 0.7$) when \II $=0.1$, $\gamma=0.05$, and $\kappa_g = 0.7$ (also when $\kappa_g \in \{0.5, 0.6\}$, not shown). This suggests a non-trivial region of low centralization where the possibility of a high infection cost (initial population is largely uninfected) drives State $2$ towards a strict policy that allows State $1$ to free-ride despite the latter including all favorite destinations. 
 
More generally, when the Government cares almost entirely about the  infection cost ($\kappa_g = 0.99$), the evidence of free-riding is drastically reduced in general, as long as there is some willingness to comply.

\subsubsection{Fairness}\label{sec:exp-fairness}
\begin{figure}[ht!]
    \centering
	\begin{subfigure}[t]{0.49\textwidth}
	\centering
		\includegraphics[scale=0.4]{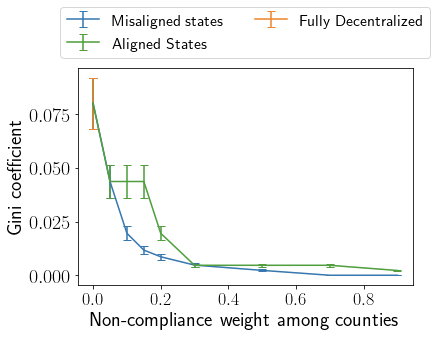}
		\caption{Symmetric transport matrix. \label{fig:synth_uniform}}
	\end{subfigure}~
	\begin{subfigure}[t]{0.49\textwidth}
	\centering
		\includegraphics[scale=0.4]{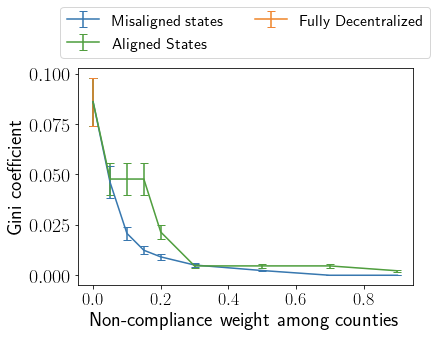}
		\caption{Transport asymmetry with one favorite County per State\label{fig:synth_favcounty}}
	\end{subfigure}
		\caption{Gini coefficient averaged over 
	30 trials for \textit{Aligned States}, 50 for each other scenario. Error bars show one standard error.}\label{fig:gini}
\end{figure}

We also study how \textit{fair} the distribution of costs is among the Counties at an HECG equilibrium for different degrees of centralization and different priorities of the States. Of the many fairness concepts that exist in the literature, we apply the popular measure, the \textit{Gini coefficient} \cite{gini1912variabilita}, to Counties' overall costs at profile $\actvec$ returned by Algorithm~\ref{alg:ne}:
\[ \mathtt{Gini}(\actvec) = \frac{\sum_{a \in \C}\sum_{a' \in \C}\lvert \Co_a(\actvec)-\Co_{a'}(\actvec) \rvert}{2n_L\sum_{a \in \C}\Co_a(\actvec)}.\]

For all these experiments, we set $\kappa_{g}=0.5$ for the Government, and $\gamma_a=0.5$ for both States. 
There are again two subsets with respect to the transport matrix $R$. The first subset is under transport symmetry (as in all the experiments in Section~\ref{sec:exp-soc-wel} and the first subset in Section~\ref{sec:exp-free-riding}), and the second under a specific transport asymmetry based on the idea of globally favorite destinations (as expounded in Section~\ref{sec:exp-free-riding}); however, for the second subset of reported experiments on fairness, we make exactly one arbitrarily chosen County in each State a globally favorite destination ($c_1$ and $c_6$), with $r^H = 0.35$, hence $\lvert \C_f \rvert = 2$ and $r^L=0.0375$.

For each of the above subsets, all Counties $a \in \C$ share their non-compliance weight $\gamma_a = \gamma$, which we vary from $0$ to $0.9$ in steps of $0.01$, and set $\kappa_a = \kappa'_a (1-\gamma_a)$ where $\kappa'_a$ are drawn i.i.d. from $\mathbb{U}[0,1]$. For each combination of weight vectors of Counties, we have three scenarios based on County compliance and States' relative weighting of infection and implementation costs:
\begin{itemize}
    \item \textit{Fully decentralized} scenario: each County's non-compliance weight is set to $0$; this corresponds to a simultaneous-move game among Counties only, and the weights of both States (and the Government) are immaterial.
    \item \textit{Misaligned States}: For State $1$, we set $\kappa_1=0.25\eta_1=0.2(1-\gamma_1)$; for State $2$, we set $\kappa_1=4\eta_1=0.8(1-\gamma_1)$.
    \item \textit{Aligned States}: we set $\kappa_a=\eta_a=0.5(1-\gamma_a)$ for each State $a\in \{1,2\}$.
\end{itemize}
Thus, unlike in Sections~\ref{sec:exp-soc-wel} and~\ref{sec:exp-free-riding}, weight vectors of Counties are not identical to those of the respective States, in general.

Each set of draws of $\kappa'_a$ for all Counties constitutes one trial.
Figure~\ref{fig:gini} reports the resulting average Gini coefficient values over all the trials (lower is better). $\mathtt{Gini}(\actvec)$ is consistently low for all experimental settings we tried, even with full decentralization (never exceeding $0.1$); however, there is a significant drop when the structure becomes hierarchical with Counties willing to comply with respective States even to a small extent (see left half of each panel of Figure~\ref{fig:gini}), regardless of the priority misalignment between these States. Thus, we have identified a (sub-)space of parameter configurations for which a 3-level hierarchy improves overall fairness among actors \textit{implementing} the policy compared to full autonomy given to them. Another notable observation is that the Gini coefficient characteristics are essentially unaffected by the variation in the transport matrix we considered in our experiments (Figures~\ref{fig:synth_uniform} and~\ref{fig:synth_favcounty}). 

\subsection{Experiments on Real Data}\label{sec:ny-nj}
For these experiments, we considered recent real data on two states in the U.S., New York (NY) and New Jersey (NJ) that have 62 and 21 counties respectively, and imagined them to form a closed system under a single Government.  We summarized data on these two States from the U.S. Department of Transportation and Census Bureau to estimate population shares and the transport matrix for this world (as we show below), and varied cost function weights  and initial infection rates in our experiments.

BRD is slow and scales poorly to this large game (83 Counties), so we approximated equilibria in this world using the QIP approach (Section~\ref{sec:qip-methods} Algorithms~\ref{alg:taylor_iter}~and~\ref{alg:qip-brd}) that trades off computation speed against accuracy in equilibrium computation. We used the following parameter settings for Algorithm~\ref{alg:qip-brd}: $\Delta_g = 0.1, T_2 = 2, k_2 = 2, e_2 = 10^{-6}, \actvec_{\C_0} = (0.5,0.5,\dots,0.5)$. However, in some instances, the QIP approach failed to find a solution; for those  instances we use BRD with a discretization factor of $0.1$ instead.

We will start by describing the real-world datasets that we used to construct the game environment for our experiments. We collected three sets of publicly available data on NY and NJ:
\begin{enumerate}[leftmargin=*,label=(\Alph*)]
    \item \label{roads} Two data sets on U.S. roads and road traffic published on September 30, 2020 by U.S. Department of Transportation, Federal Highway Administration (FHWA):
    \begin{enumerate}[leftmargin=*,label=(\arabic*)]
        \item \label{roadlength} \textbf{HM-20M 2019}\footnote{\url{https://www-fhwa-dot-gov.proxy.lib.umich.edu/policyinformation/statistics/2019/hm20m.cfm}}: This set tabulates Public Road Length (as of 2019) in kilometers for every U.S. state (rows), roads being categorized into RURAL and URBAN, and further sub-categorized into INTERSTATE, OTHER FREEWAYS AND EXPRESSWAYS, OTHER PRINCIPAL ARTERIAL, MINOR ARTERIAL, MAJOR COLLECTOR, MINOR COLLECTOR, and LOCAL (columns). These categories are called \textit{functional systems}. We  only use data from the rows corresponding to the states of NY and NJ.
        \item \label{traffic} \textbf{VM-2 2019}\footnote{\url{https://www-fhwa-dot-gov.proxy.lib.umich.edu/policyinformation/statistics/2019/vm2.cfm}}: This set tabulates a quantification of Functional System Travel (as of 2019) in Annual Vehicle-Miles for each state and for each of the same categories and sub-categories of roads as \textbf{HM-20M 2019}. Again, we only use data from the rows corresponding to NY and NJ.
    \end{enumerate}
    \item \label{countypops} \textbf{County Population Totals} (2010-2019)\footnote{\url{https://www-census-gov.proxy.lib.umich.edu/data/datasets/time-series/demo/popest/2010s-counties-total.html}}: A dataset released in March, 2020 by the Population division of U.S. Census Bureau. For each state, a table records annual estimates of the resident population for counties in each state as of July 1 every year (rows) over the years 2010-2019 (columns), based on April 1 2010 U.S. census data. We use data for all 62 counties of NY\footnote{Annual Estimates of the Resident Population for Counties in New York: April 1, 2010 to July 1, 2019 (CO-EST2019-ANNRES-36); Source: U.S. Census Bureau, Population Division; Release Date: March 2020} and all 21 counties NJ\footnote{Annual Estimates of the Resident Population for Counties in New Jersey: April 1, 2010 to July 1, 2019 (CO-EST2019-ANNRES-34); Source: U.S. Census Bureau, Population Division; Release Date: March 2020} for the year 2019 only.
\end{enumerate}

We ran our experiments on an HECG with one federal government $g$, two States NY ($s_1$) and NJ ($s_2$) only, with 62 and 21 arbitrarily numbered Counties $\{c_1,c_2,\dots,c_{62}\}$ and $\{c_{63},c_{64},\dots,c_{83}\}$ under them respectively (see Figure~\ref{fig:model}).
For each County $c$, the share $\mu_c$ as defined in  Section~\ref{sec:game_model} is naturally given by the ratio of the population of County $c$ to the total population of all 83 Counties in the two States, all directly available from the data set~\ref{countypops}. Hence, we can obtain the shares of the States as  $\mu_{s_1}=\sum_{i=1}^{62} \mu_{c_i}$, and $\mu_{s_2}=\sum_{i=63}^{83} \mu_{c_i}.$ 
We also use all three datasets reported above to estimate the transport matrix $R = \{r_{cc'}\}_{c,c' \in \C}$ defined in Section~\ref{sec:infec}, as explained in the rest of this section. No other component of the model is estimated from these datasets.

First, we compute the \textit{traffic} on each road category in datasets~\ref{roads} as the ratio of vehicle-miles (from \textbf{VM-2 2019}) to road length (from \textbf{HM-20M 2019}). Then, for each State $s \in \{s_1,s_2\}$, we compute 3 intermediate quantities to use in transport matrix entries:
\begin{itemize}[leftmargin=*]
    \item $P_{\text{in-county}}(s)$, the ratio of the total traffic on the MAJOR COLLECTOR, MINOR COLLECTOR, and LOCAL road categories in State $s$ to the total traffic on all road categories in State $s$;
    \item $P_{\text{in-state}}(s)$, the ratio of the total traffic on the OTHER FREEWAYS AND EXPRESSWAYS, OTHER PRINCIPAL ARTERIAL, and MINOR ARTERIAL road categories in State $s$ to the total traffic on all road categories in State $s$;
    \item $P_{\text{between-states}}(s)$,  the ratio of the total traffic on the INTERSTATE road category in State $s$ to the total traffic on all road categories in State $s$.
\end{itemize}
Finally, we have
\begin{align*}
    r_{cc'} = \begin{cases}
    P_{\text{in-county}}(s) &\text{if $c = c'$},\\
    \frac{\mu_{c'}}{\mu_s} \cdot P_{\text{in-state}}(s) &\text{if $c\neq c'$, $c \in \chi(s)$, $c' \in \chi(s)$},\\
    \frac{\mu_{c'}}{1-\mu_s} \cdot P_{\text{between-states}}(s) &\text{if $c \in \chi(s)$, $c' \not\in \chi(s)$},
    \end{cases}
\end{align*}
where $\chi(s)$ denote the Counties in State $s$. We now report the results of the three sets of experiments that we performed in this environment, analogous to those in Section~\ref{sec:expt_synth}.

\subsubsection{Social welfare}\label{sec:NYNJ-exp-soc-wel}
As in Section~\ref{sec:exp-soc-wel}, we measure the social cost as the Government's overall cost $\Co_{g}(\actvec) = \kappa_{g} \Cinc_{g}(\actvec) + (1-\kappa_{g})\Cdec_{g}(\actvec)$, where we vary $\kappa_g$ in $(0,1)$. Here, we examine three of the four centralization scenarios introduced in Section~\ref{sec:exp-soc-wel}: \CCS, \CU, and \EQthree. For \EQthree, all States and Counties share the non-compliance cost weight $\gamma$, which we vary from $0$ to $1$ in steps of $0.1$ and set $\kappa_a = 0.9(1-\gamma)$ for every $a \in \S \cup \C$. We fix the initial infection rate of every County in NJ at $0.1$, and vary that of every County in NY (all equal) from $0.1$ to $0.9$ in steps of $0.1$, the latter being referred to as \II. 
For \CCS and \CU, we used projected stochastic gradient descent (PGD) for social cost optimization with the learning rate set at $0.2$ and number of iterations at $10^4$, as in Section~\ref{sec:exp-soc-wel}.
\begin{figure}[ht!]
    \centering
	\begin{subfigure}[t]{0.45\textwidth}
		\includegraphics[scale=0.31]{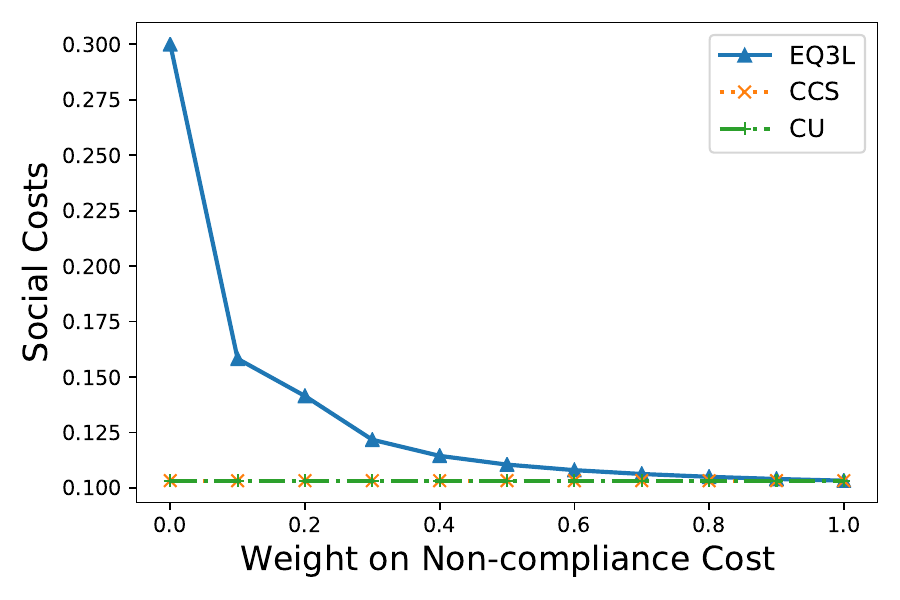}
		\caption{\II $=0.7$, $\kappa_g=0.7$ \label{fig:NYNJ_soc_cost_gamma1}}
	\end{subfigure}~
	\begin{subfigure}[t]{0.45\textwidth}
		\includegraphics[scale=0.31]{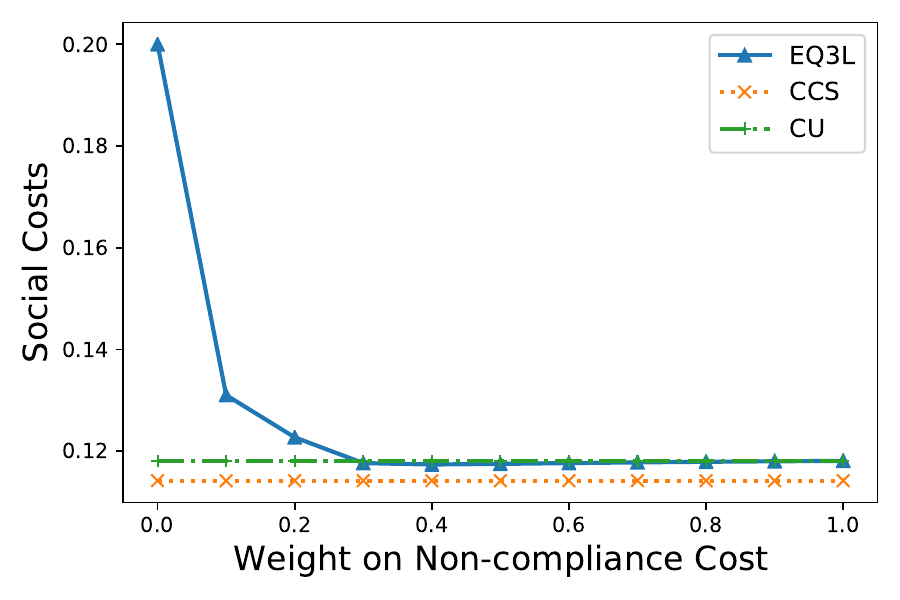}
		\caption{\II $=0.7$, $\kappa_g=0.8$ \label{fig:NYNJ_soc_cost_gamma2}}
	\end{subfigure}\\
	\begin{subfigure}[t]{0.45\textwidth}
		\includegraphics[scale=0.31]{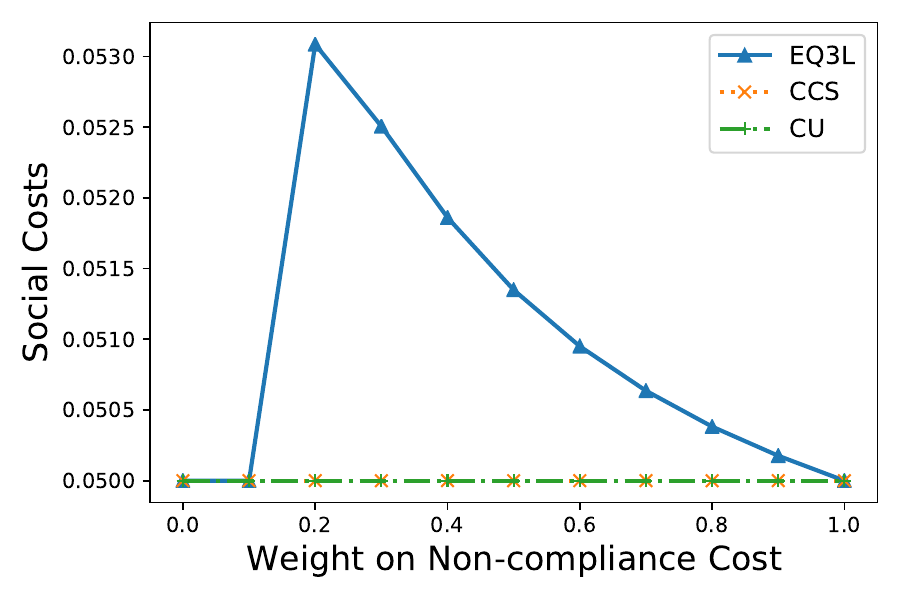}
		\caption{\II $=0.7$, $\kappa_g=0.95$ \label{fig:NYNJ_soc_cost_gamma3}}
	\end{subfigure}~
	\begin{subfigure}[t]{0.45\textwidth}
		\includegraphics[scale=0.31]{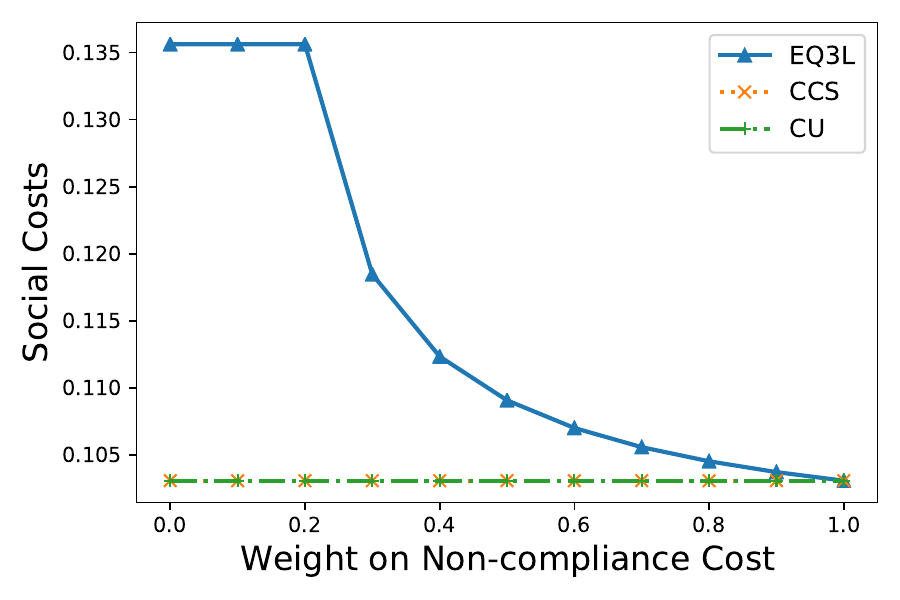}
		\caption{\II $=0.8$, $\kappa_g=0.7$ \label{fig:NYNJ_soc_cost_gamma4}}
	\end{subfigure}\\
	\begin{subfigure}[t]{0.45\textwidth}
		\includegraphics[scale=0.31]{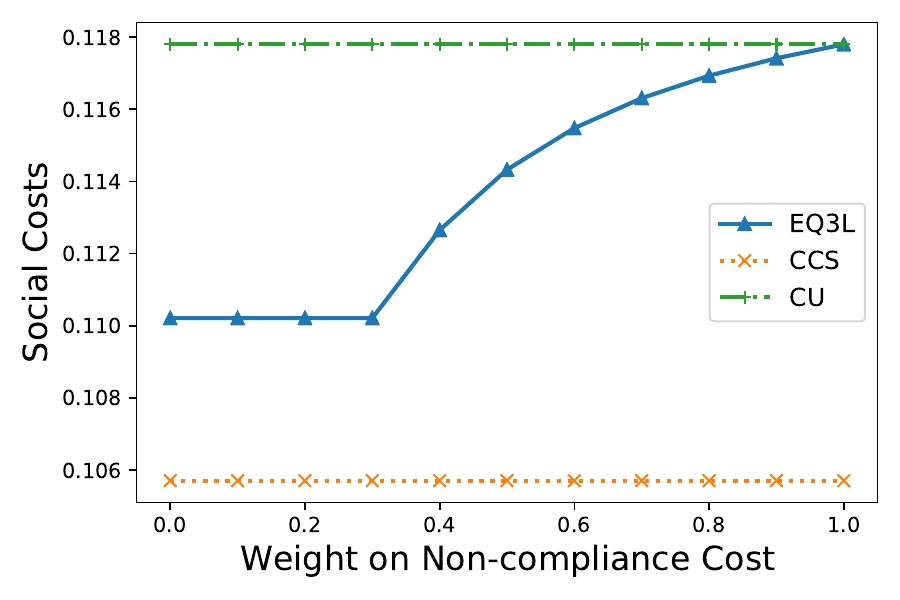}
		\caption{\II $=0.8$, $\kappa_g=0.8$ \label{fig:NYNJ_soc_cost_gamma5}}
	\end{subfigure}~
	\begin{subfigure}[t]{0.45\textwidth}
		\includegraphics[scale=0.31]{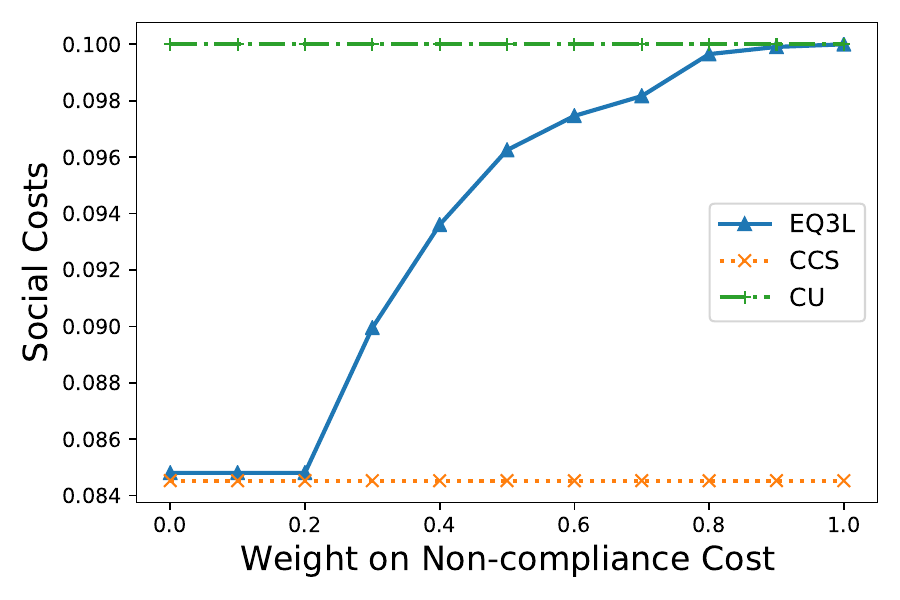}
		\caption{\II $=0.8$, $\kappa_g=0.9$ \label{fig:NYNJ_soc_cost_gamma6}}
	\end{subfigure}\\
	\begin{subfigure}[t]{0.45\textwidth}
		\includegraphics[scale=0.31]{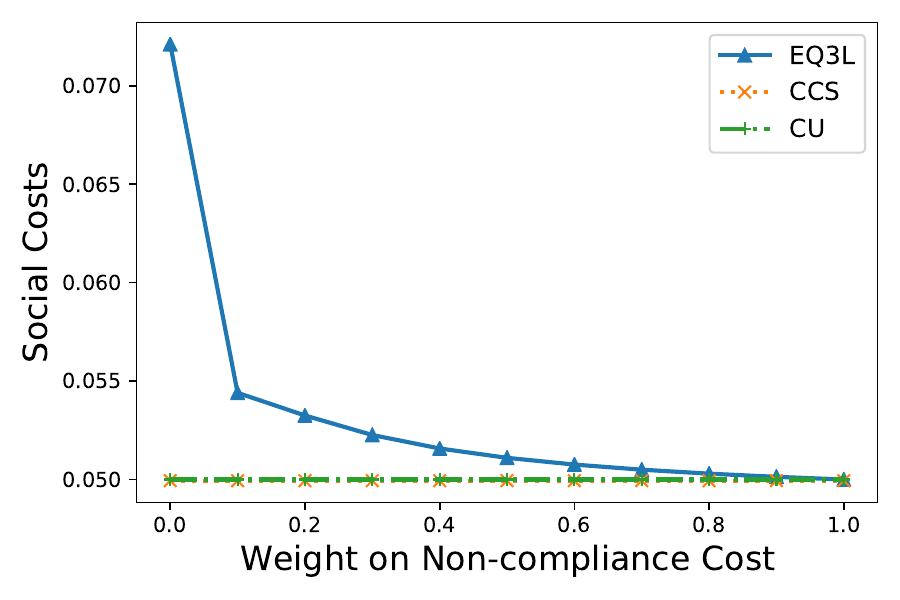}
		\caption{\II $=0.8$, $\kappa_g=0.95$ \label{fig:NYNJ_soc_cost_gamma7}}
	\end{subfigure}
	\caption{Social cost (y-axis) for each of our centralization scenarios as a function of the shared non-compliance cost weight (x-axis), with \II and $\kappa_g$ as parameters, from experiments based on NY and NJ data. Error bars show one standard error (barely visible)}
	\label{fig:NYNJ_soc_cost_gamma}
\end{figure}

The broad characteristics are similar to those of the corresponding experimental results for synthetic data (Section~\ref{sec:exp-soc-wel} Figure~\ref{fig:soc_cost_gamma}). For values of $\kappa_g$ outside a small sub-interval of large values, \EQthree social costs are at least as high as those under the centralized scenarios for all values of \II considered, indicating that there is no advantage to decentralization in terms of social cost; however, within the said interval, there are values of \II (initial infection rate of NJ) for which the decentralized social cost is (weakly) lower than that of the uniform centralized solution (Figures~\ref{fig:NYNJ_soc_cost_gamma5} and~\ref{fig:NYNJ_soc_cost_gamma6}).

\begin{figure}[ht!]
\centering
	\begin{subfigure}[t]{0.45\textwidth}
 		\includegraphics[scale=0.35]{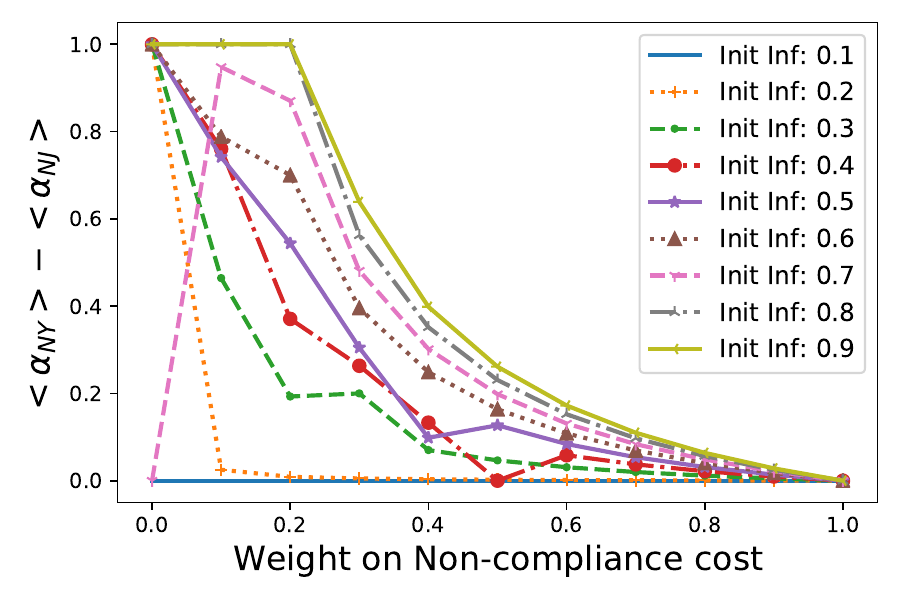}
		\caption{$\kappa_g=0$ \label{fig:fr_nynj_0}}
	\end{subfigure}~
	\begin{subfigure}[t]{0.45\textwidth}
 		\includegraphics[scale=0.35]{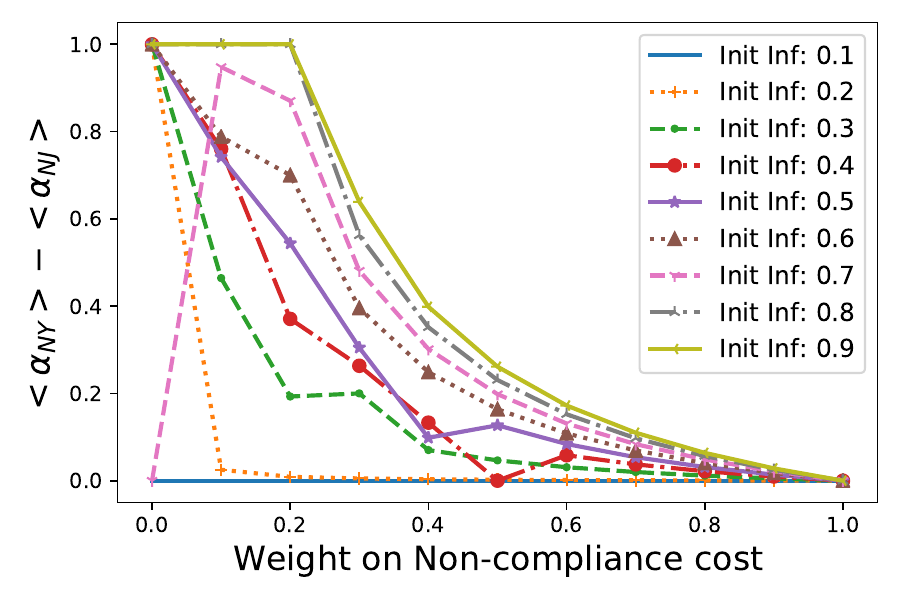}
		\caption{$\kappa_g=0.2$ \label{fig:fr_nynj_p2}}
	\end{subfigure}\\
	\begin{subfigure}[t]{0.45\textwidth}
 		\includegraphics[scale=0.35]{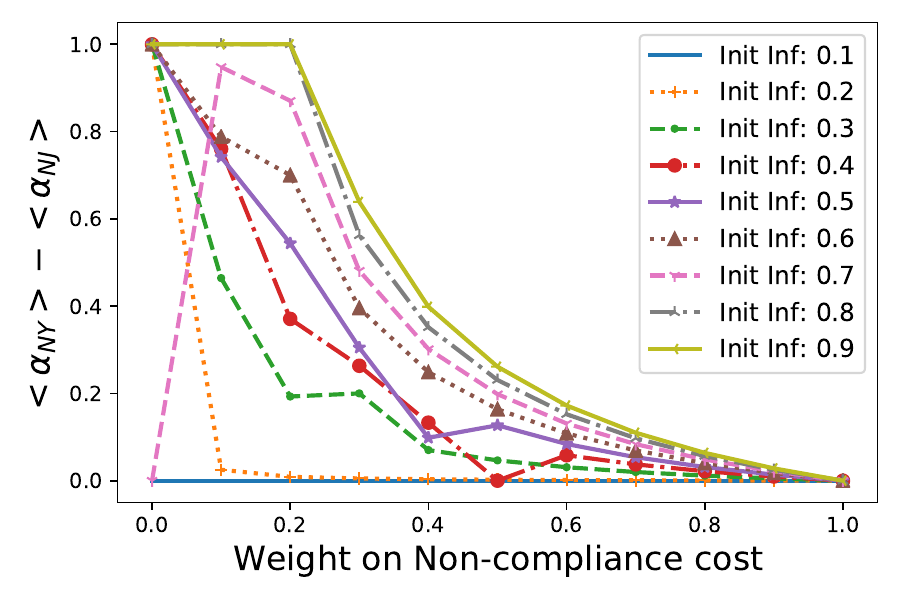}
		\caption{$\kappa_g=0.5$ \label{fig:fr_nynj_p5}}
	\end{subfigure}~
	\begin{subfigure}[t]{0.45\textwidth}
 		\includegraphics[scale=0.35]{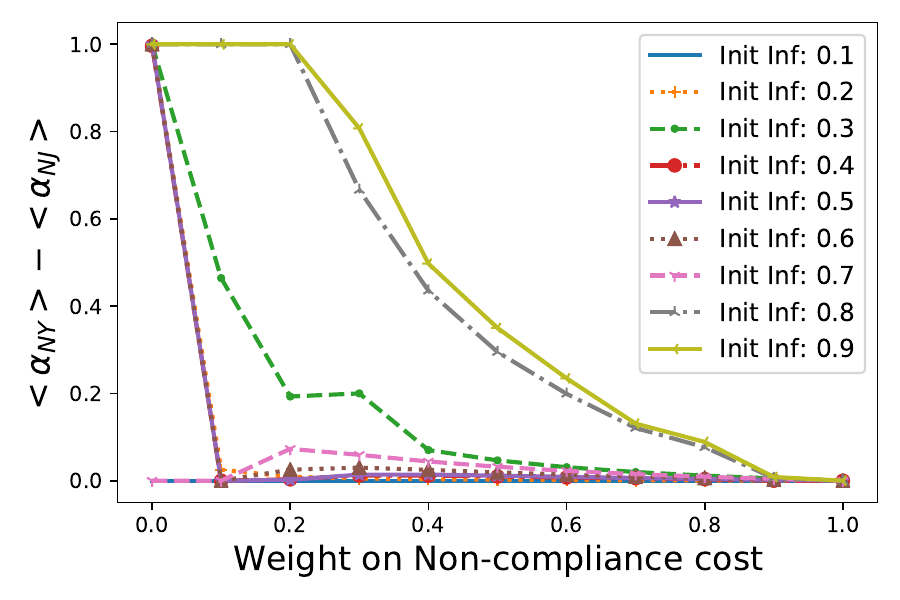}
		\caption{$\kappa_g=0.9$ \label{fig:fr_nynj_p9}}
	\end{subfigure}\\
	\begin{subfigure}[t]{0.45\textwidth}
 		\includegraphics[scale=0.35]{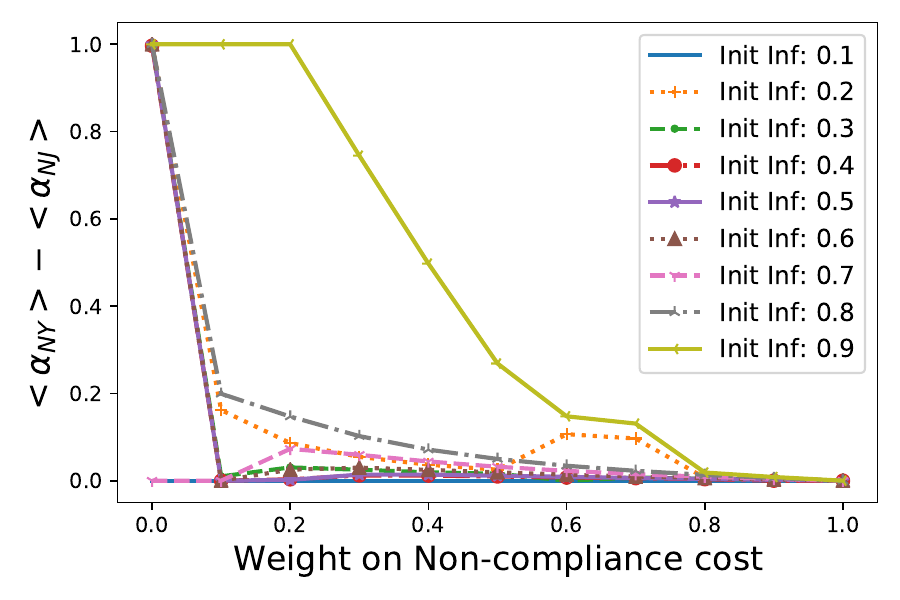}
		\caption{$\kappa_g=0.95$ \label{fig:fr_nynj_p95}}
	\end{subfigure}~
	\begin{subfigure}[t]{0.45\textwidth}
 		\includegraphics[scale=0.35]{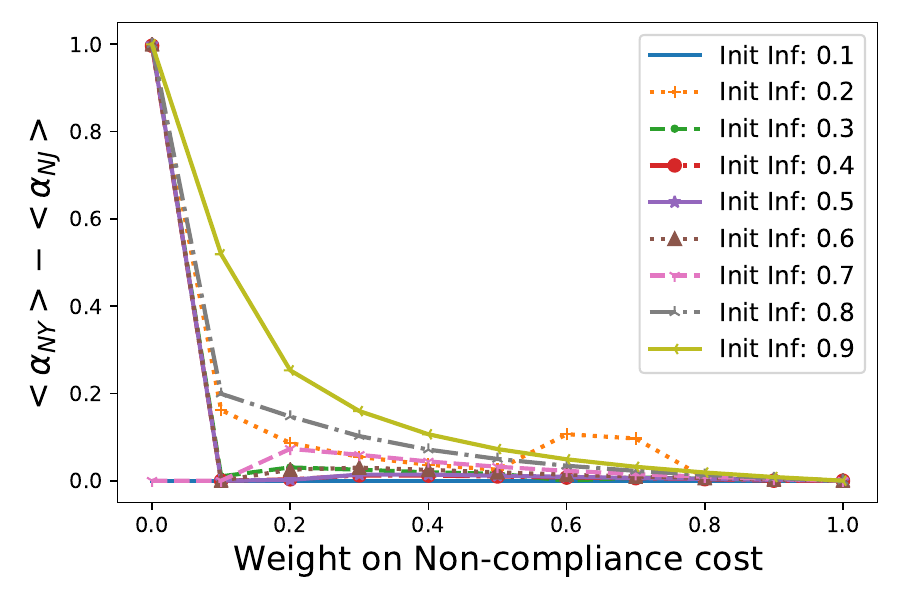}
		\caption{$\kappa_g=1$ \label{fig:fr_nynj_1}}
	\end{subfigure}
	\caption{Free-riding (y-axis) by NY off NJ in our model, as a function of non-compliance weight (x-axis). Results are averaged over $10$ BRD iterations for each configuration; error bar shows one standard error (barely visible).}
	\label{fig:freeride_nynj}
\end{figure}

\subsubsection{Free-riding}\label{sec:NYNJ-free-riding}
We use the same proxy for free-riding as in Section~\ref{sec:exp-soc-wel}. We focus on the \EQthree scenario only, and use the same combinations of weight vectors and initial infection rates as in Section~\ref{sec:exp-soc-wel}.
Figure~\ref{fig:freeride_nynj} shows the difference in policies averaged over respective Counties, $\langle\alpha_{NY} \rangle-\langle\alpha_{NJ} \rangle$, for representative values of $\kappa_g$. For all values of $\kappa_g \in [0,1]$ that we tried, we found evidence of free-riding by NY, the higher-infection State in our model, for lower values of the universally shared non-compliance cost weight; as expected, the evidence grows weaker, in general, with higher centralization and smaller \II (NY).

\subsubsection{Fairness.}\label{sec:NYNJ-fairness}
As in Section~\ref{sec:exp-fairness}, we use the Gini coefficient of the overall costs of all $83$ Counties under NY and NJ as our measure of (un)fairness. Figure~\ref{fig:gini_NYNJ} shows our results for the initial infection rate at each County set at $0.2$ and the same three scenarios and weight configurations (including $\kappa_{g}=0.5$ for the Government) as in Section~\ref{sec:exp-fairness}. The characteristics are virtually identical to those in Figures~\ref{fig:synth_uniform} and~\ref{fig:synth_favcounty}: unfairness is already low even under full decentralization, and drops quickly with more centralization regardless of State-level priorities.  
\begin{figure}[ht!]
    \centering
	\centering
		\includegraphics[scale=0.4]{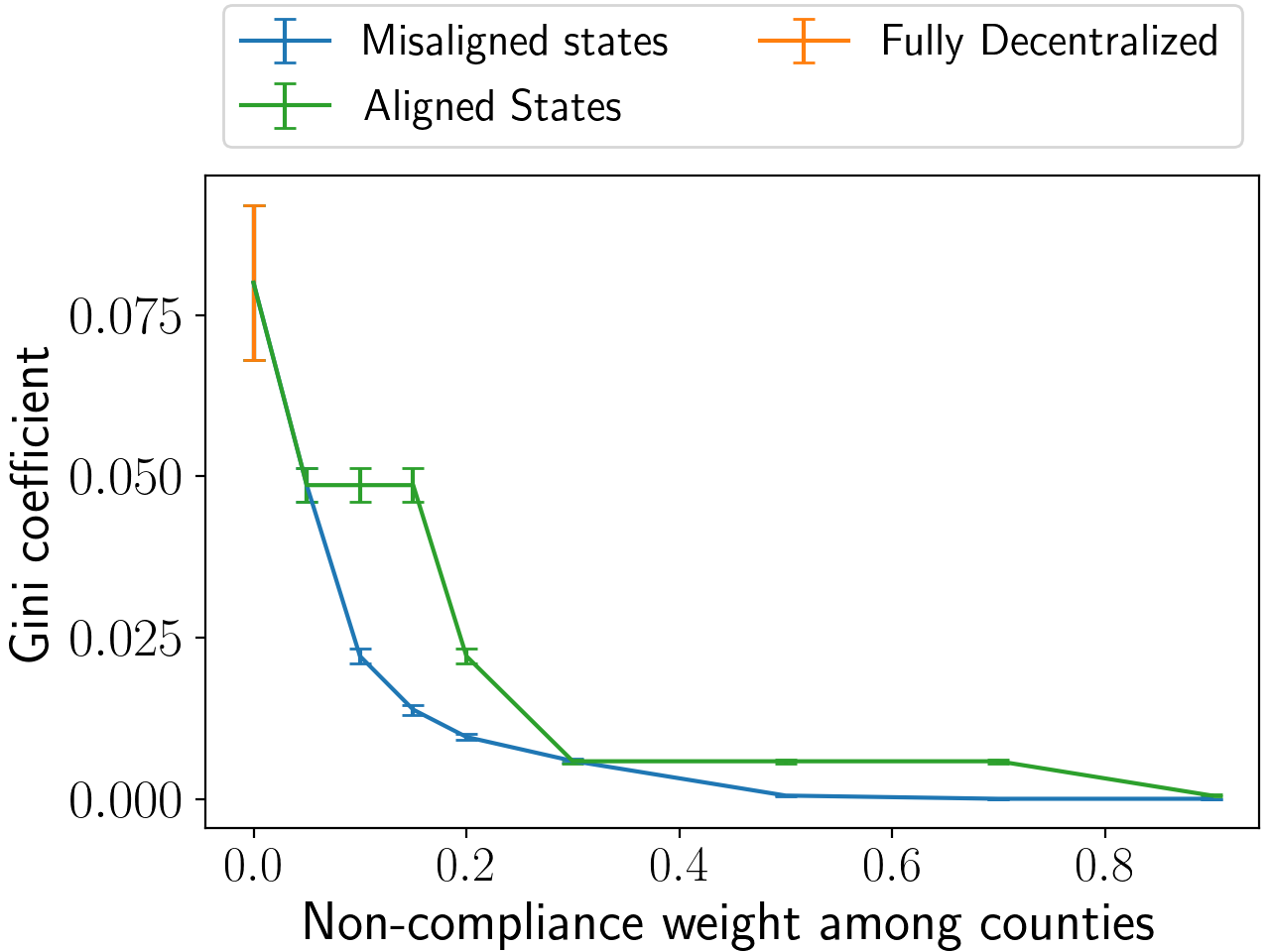}
		\caption{Gini coefficient averaged over 
	30 trials for \textit{Aligned States}, 50 for each other scenario. Error bars show one standard error.}\label{fig:gini_NYNJ}
\end{figure}

\section{Future Work}\label{sec:disc}
Further research directions include more extensive experimentation under various parameter configurations (e.g. testing causal hypotheses); using the actual ABM \cite{wilder2020modeling} instead of our approximations and handling the resulting computational issues; considering more complex policies (e.g. adaptive strategies or multi-dimensional action spaces) and invoking more sophisticated empirical game-theoretic analysis; applying HECG or its natural variants to other problems of hierarchical decision-making such as within a corporation or political organization, or in the interplay of domestic and international politics \cite{putnam1988diplomacy}.





\bibliographystyle{ACM-Reference-Format} 
\bibliography{abb,covid}


\appendix
\section{Omitted Details from Section~2.3}\label{app:closedform}
We will state and prove a useful property of Poisson distributions.
\begin{proposition}\label{prop:exppoisson}
Consider a random variable $Y \sim \Poi(\lambda)$. Then, for any non-zero real number $b$ independent of $Z$,
\[\mathbb{E}_Z [b^Z] = e^{-\lambda(1-b)}.\]
\end{proposition}
\begin{proof}
From definitions,
\begin{align*}
    \mathbb{E}_Z [b^Z] &= \sum_{z=0}^\infty a^z \Pr[Z=z] =   \sum_{z=0} a^z \cdot \frac{e^{-\lambda} \lambda^z}{z!} = e^{-\lambda} \sum_{z=0}^\infty  \frac{(b\lambda)^z}{z!} = e^{-\lambda} e^{b\lambda},
\end{align*}
which equals the desired expression.
\end{proof}

\section{Omitted Details from Section~3.2}
\label{sec:methods-qip-appx}
In this section, we provide 
additional experimental results comparing BRD and Taylor-Iter-BR. In Figure~\ref{fig:1_10}, we show the performance of Taylor-Iter-BR (Algorithm~\ref{alg:taylor_iter}) on an HECG having only one State with a non-compliance weight of $0$ and 10 Counties $\{c_1,c_2,\dots,c_{10}\}$ under it, i.e. effectively a 2-level game with 10 agents on the lower level. The State has a high infection rate $\kappa_s = 0.9$, and the Counties share the vector of weights $\kappa_c = \eta_c = \gamma_c = 1/3, \ \forall c\in \C$. 
Recall that an approximation to any function based on its Taylor series expansion is sensitive to the input point about which this expansion is taken, and the choice of this point is non-trivial.
For our Taylor series-based cost function approximation, we expand the infection function 
at a point $(\alpha_{\C_0},\alpha_{\C_0},\dots,\alpha_{\C_0})$ (the length of the vector being the number of counties $n$) where $\alpha_{\C_0} \in \{0.1,0.2,\dots,0.9\}$. 

We introduce asymmetry into the scenario by setting the initial infection rate in County $c_i$ at $0.8/(i+1)$ for $i=1,2,\dots,10$. As a standard for comparison, we computed the approximate equilibrium of this game with BRD$(0.01)$, i.e the BRD method (Section~\ref{sec:brd} Algorithm~\ref{alg:ne}) with discretization factor $0.01$; the chosen point for the Taylor series expansion that comes closest to the BRD solution corresponds to $\alpha_{\C_0}=\alpha^*_{\C_0}=0.5$.
 
 \begin{figure}[ht!]
	\centering
	\begin{subfigure}[t]{0.49\linewidth}
	\includegraphics[width=0.99\linewidth]{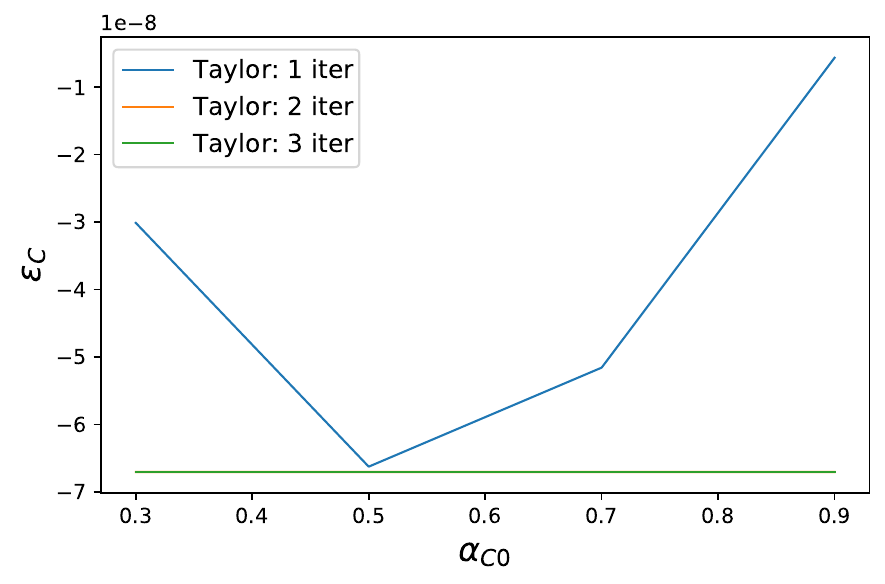}
	\end{subfigure}
    \begin{subfigure}[t]{0.49\linewidth}
	\includegraphics[width=0.99\linewidth]{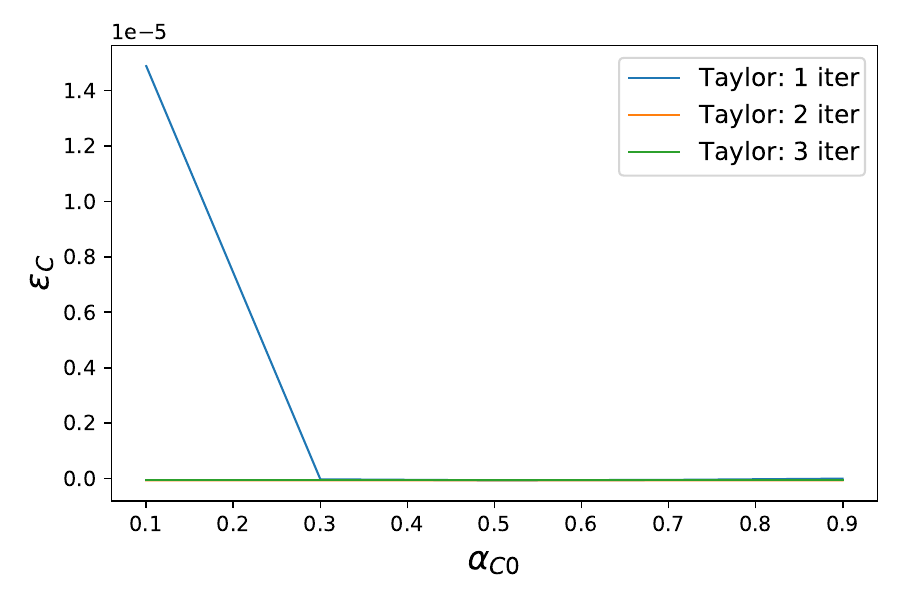}
	\end{subfigure}
	\caption{The performance of Taylor-Iter-BR on an HECG with 1 State and 10 Counties.}
	\label{fig:1_10}
\end{figure}

We observe that (1) the closer the actual point of expansion is to $(\alpha^*_{\C_0},\alpha^*_{\C_0},\dots,\alpha^*_{\C_0})$, the lower the $\epsilon$ (suggesting better performance). (2) 2 iterations of the Taylor-Iter-BR can achieve good performance regardless of the points of expansion. Although it is a heuristic algorithm, this phenomenon can be witnessed in other settings. (3) The $\epsilon$ of the counties are negative (2 and 3 iterations), which means that the algorithm finds a better approximation to the equilibrium than BRD search.

\begin{figure}[ht!]
	\centering
	\begin{subfigure}[t]{0.49\linewidth}
	\includegraphics[width=0.99\linewidth]{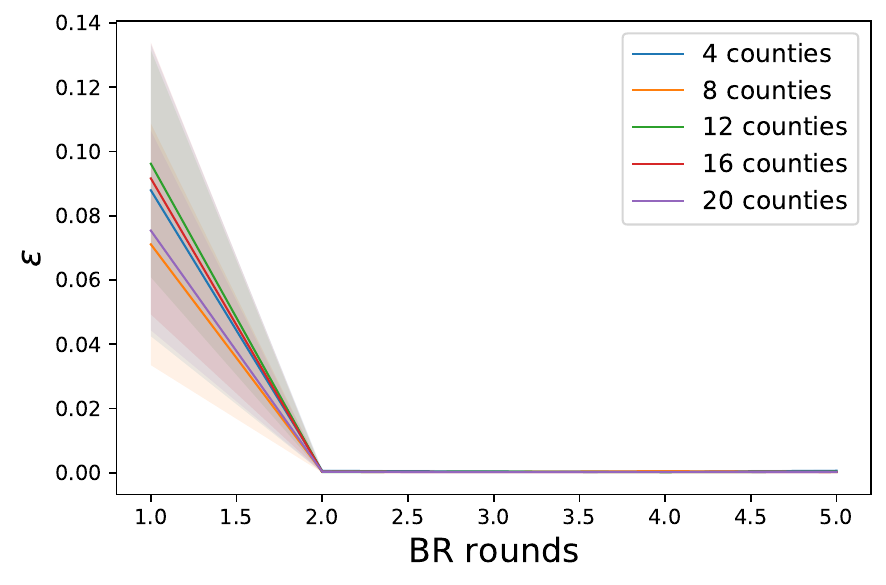}
	\end{subfigure}~
    \begin{subfigure}[t]{0.49\linewidth}
	\includegraphics[width=0.99\linewidth]{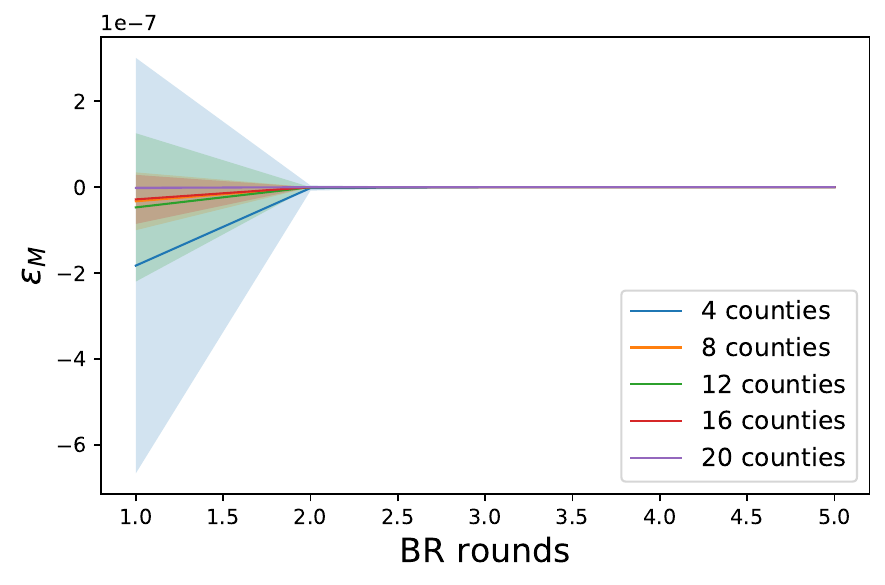}
    \end{subfigure}
	\caption{The performance of Taylor-Iter-BR of 2 State and $\C$ counties for varying $\C$.}
	\label{fig:taylor-iter-br}
\end{figure}

Then we vary the number of counties from 5 to 35 to compare the performance and run-time. We approximate the infection at point $0.5$ and conduct two iterations of the Taylor-Iter-BR. The initial infection rate of the first $N/2$ Counties is set to 0.7 while others are set to 0.2. We compare the Taylor-Iter-BR and BRD methods with different parameters. As shown in Figure \ref{fig:QIP}, the run-time of Taylor-Iter-BR (2 iterations) is more efficient than the binary searched BRD with discretization factor 0.01 and 0.05 but performs better than BRD (0.05) and as well as BRD (0.01). BRD (0.1) is efficient but selects the strategies with the highest $\epsilon$. All the Taylor-Iter-BR (1,2,3 iterations) perform better than BRD (0.01) since they have negative $\epsilon$.

After testing the 1 State N Counties cases, we also demonstrate the performance of the state's best response dynamics using Taylor-Iter-BR (2 iterations) under 2 States N Counties cases. $N$ is from 4 to 20. Each state has for $N/2$ Counties under it. The Counties of State 1 have their initial infection rates sampled from $\mathbb{U}[0.5,1]$ whereas those of the others are sampled from $\mathbb{U}[0,0.5]$. BRD finds the pure strategy NE in $2$ best response rounds.

\end{document}

